%% file: main-fop.tex
\documentclass[11pt]{article}
\def\withcolors{0}
\def\withnotes{0}
\input{packages}
\input{preamble}

\def\authornameob{Omri Ben-Eliezer}
\def\authoraffiob{Blavatnik School of Computer Science, Tel-Aviv University. Email: \email{omrib@mail.tau.ac.il}.}
\def\authornamecc{Cl\'{e}ment L. Canonne}
\def\authorafficc{Columbia University. Email: \email{ccanonne@cs.columbia.edu}. Research supported by NSF grants CCF-1115703 and NSF CCF-1319788.}

\title{Improved Bounds for Testing Forbidden Order Patterns}
\date{\today}
 
\author{
  \authornameob\thanks{\authoraffiob}
  \and \authornamecc\thanks{\authorafficc}
}

\begin{document}

\maketitle

\pagenumbering{gobble}
\begin{abstract}
    \input{abstract}
\end{abstract}

\ifnum\withnotes=1
  \clearpage
  \listoftodos
  \hfill
  \newpage
  \tableofcontents
  \clearpage
\fi

\clearpage\pagenumbering{arabic}

\section{Introduction}\label{sec:intro}
\input{sec-introduction}

\section{Previous work}\label{sec:previous}
\input{sec-previous}

\section{Our contributions}\label{sec:contributions}
\input{sec-contributions}

\section{Discussion and open problems}\label{sec:discussion}
\input{sec-discussion}

\section{Upper bound}\label{sec:upper}
\input{sec-upper}

\section{Lower bounds}\label{sec:lower}
\input{sec-lower}

\section{Hierarchy of adaptivity}\label{sec:hierarchy}
\input{sec-hierarchy}

\section{A permutation dependent non-adaptive hierarchy}\label{sec:other}
\input{sec-other}

\section{Acknowledgments}
The authors wish to thank Noga Alon for fruitful discussions and invaluable feedback.
\clearpage
\bibliographystyle{alpha}
\bibliography{references} 

\clearpage
\appendix

\end{document}

%% file: packages.tex
\usepackage[T1]{fontenc}
\usepackage[utf8]{inputenc}

\usepackage{lmodern}
\usepackage{xspace}                                     

\usepackage{amsfonts,amsmath,amssymb, amsthm, mathtools}
\usepackage{thm-restate}
\usepackage{dsfont} 

\usepackage{algorithmicx,algpseudocode,algorithm}

\usepackage[usenames,dvipsnames,table]{xcolor}

\usepackage{tikz}
\usetikzlibrary{arrows}
\usetikzlibrary{calc,decorations.pathmorphing,patterns}

\usepackage[backref,colorlinks,citecolor=blue,bookmarks=true,linktocpage]{hyperref}
\usepackage{aliascnt}
\usepackage[numbered]{bookmark}

\usepackage{fullpage}

\usepackage[shortlabels]{enumitem}
  \setitemize{noitemsep,topsep=3pt,parsep=2pt,partopsep=2pt} 
  \setenumerate{itemsep=1pt,topsep=2pt,parsep=2pt,partopsep=2pt}
  \setdescription{itemsep=1pt}
  
\ifnum\withnotes=1
  \usepackage[colorinlistoftodos,textsize=scriptsize]{todonotes}
\fi

\usepackage{mleftright} 

\usepackage{pgfplots}

%% file: abstract.tex
A sequence $f\colon\{1,\dots,n\}\to\mathbb{R}$ contains a permutation $\pi$ of length $k$ if there exist $i_1<\dots<i_k$ such that, for all $x,y$, $f(i_x)<f(i_y)$ if and only if $\pi(x)<\pi(y)$; otherwise, $f$ is said to be $\pi$-free. In this work, we consider the problem of testing for $\pi$-freeness with one-sided error, continuing the investigation of [Newman et al., SODA'17].

We demonstrate a surprising behavior for non-adaptive tests with 
one-sided error: While a trivial sampling-based approach yields an $\varepsilon$-test for $\pi$-freeness making 
$\Theta(\varepsilon^{-1/k} n^{1-1/k})$ queries, our lower bounds imply that this is almost optimal for most permutations! Specifically, for most permutations $\pi$ of length $k$, any non-adaptive one-sided $\varepsilon$-test requires 
$\varepsilon^{-1/(k-\Theta(1))}n^{1-1/(k-\Theta(1))}$ 
queries; furthermore, the permutations that are hardest to test require $\Theta(\varepsilon^{-1/(k-1)}n^{1-1/(k-1)})$ queries, which is tight in $n$ and $\varepsilon$.

Additionally, we show two hierarchical behaviors here. First, for any $k$ and $l\leq k-1$, there exists some $\pi$ of length $k$ that requires $\tilde{\Theta}_{\varepsilon}(n^{1-1/l})$ non-adaptive queries. Second, we show an adaptivity hierarchy for $\pi=(1,3,2)$ by proving upper and lower bounds for (one- and two-sided) testing of $\pi$-freeness with $r$ rounds of adaptivity. The results answer open questions of Newman et al. and [Canonne and Gur, CCC'17].

%% file: sec-introduction.tex
\subsection{Background and motivation}
Introduced by Rubinfeld and Sudan~\cite{RS:96} and Goldreich, Goldwasser, and Ron~\cite{GGR:98}, the field of property testing is concerned with obtaining query- and time-efficient randomized algorithms (usually running in sublinear time), which decide whether their input satisfies some prescribed property of interest, or differs significantly from any object with this property. Originating from the early work on Probabilistically Checkable Proofs, property testing quickly evolved to become an area of study of its own right, encompassing a variety of topics, including (but not limited to) testing of graphs, Boolean and real-valued functions, and probability distributions (see e.g.~\cite{Gol:10,Gol:17,BY:17,Ron:08,Ron:09,Canonne:15:Survey} for recent books and surveys).

In this work, we focus on testing properties of real-valued functions, of the form $f\colon[n]\to\R$, i.e., \emph{real-valued sequences}. A significant body of work has been dedicated to such properties, with the examples of monotonicity~\cite{Ergun:00,BGJRW:12,Dodis:99,Fischer2004,CS:2013,Raskhodnikova:14}, the Lipschitz property~\cite{JhaR:13,CS:2013}, convexity~\cite{PRR:03}, and $k$-monotonicity~\cite{CGGKW:17,GKW:17}. Very recently, Newman, Rabinovich, Rajendraprasad, and Sohler~\cite{NewmanRabinovich2017} initiated the study of a massively parameterized property,\footnote{In the massively parameterized framework, the property to be tested depends on an underlying structure, typically of a significant size, which is considered ``fixed'' (i.e., not part of the input itself). See e.g.~\cite{Newman:10} for a survey on testing massively parameterized properties.}{} \emph{order pattern freeness}, which generalizes and subsumes some of the aforementioned properties as special cases. We refer the reader to~\cite{NewmanRabinovich2017} for a discussion of several motivations for testing order pattern freeness, stemming e.g. from combinatorics and time series analysis.

In this paper, we continue the investigation of testing order pattern freeness. The problem is formulated as follows. The \emph{forbidden order pattern} is a permutation $\pi=(\pi_1,\dots,\pi_k)$ of $[k]$, viewed here as a sequence of length $k$. 
Two sequences $x = (x_1, \ldots, x_k)$ and $y = (y_1, \ldots, y_k)$ are \emph{order-isomorphic} if for any $i \neq j$, $x_i < x_j$ holds if and only if $y_i < y_j$. That is, if the relative order of the elements in both sequences is the same.
For the above permutation $\pi$, a sequence $f\colon [n]\to\R$ is said to \emph{contain the pattern $\pi$} if $\pi$ is order-isomorphic to a subsequence of $f$. In other words, $f$ contains $\pi$ if there exist $k$ integers $i_1< \dots < i_k \in [n]$ such that $f(i_a) < f(i_b)$ if and only if $\pi(a)<\pi(b)$. Accordingly, $f$ is \emph{$\pi$-free} if it does not contain the pattern $\pi$, and \emph{$\eps$-far from being $\pi$-free} if one needs to modify at least $\eps n$ of its values to make it $\pi$-free. 

An $\eps$-\emph{test} for order pattern freeness is an algorithm which, given query access to an unknown input sequence $f \colon [n] \to \R$, must distinguish (with good probability) between the case that the input is $\pi$-free and the case that it is $\eps$-far from being $\pi$-free.
As an example for $k=2$, $(2,1)$-freeness is equivalent to being monotone non-decreasing, and testing $(2,1)$-freeness amounts to testing monotonicity.

\subsection{Preliminaries and notation}\label{sec:preliminaries}
\input{sec-preliminaries}

\subsection{Organization of the results}\label{subsec:organization}
In~\autoref{sec:previous} we present the results of Newman et al.\@ \cite{NewmanRabinovich2017} on the problem of testing $\pi$-freeness. We then provide our results in~\autoref{sec:contributions}. Most of our results are in the non-adaptive case, and seem to yield a relatively good general understanding of this case.
All results in~\cite{NewmanRabinovich2017} only consider one-sided testing, and we also mainly follow this paradigm.

In~\autoref{subsec:hierarchy_adaptivity_132}, we turn to the following closely related question: How many queries are needed for a (one-sided or two-sided) test for $\pi$-freeness when we have multiple \emph{rounds} of adaptivity (as defined in~\cite{CG:17}; that is, each round is non-adaptive in itself, but after each round the test can make adaptive decisions)? We show that $\pi$-freeness for $\pi = (1, 3, 2)$ has an adaptivity hierarchy. That is, adding more rounds improves (asymptotically) the query complexity. This is the first known example of a natural property exhibiting an adaptivity hierarchy.

%% file: sec-preliminaries.tex
We hereafter denote by $[n]$ the set of integers $\{1,\dots,n\}$, and by $\log$ the logarithm in base two. 
 As aforementioned, we consider the set $\mathcal{F}$ of real-valued functions of the form $f\colon[n]\to\R$ (which we will equivalently refer to as \emph{sequences}), equipped with the usual Hamming distance: $\dist{f}{g} = \abs{ \setOfSuchThat{i\in [n]}{f(i)\neq g(i)} }$. Given a property $\property \subseteq \mathcal{F}$, the distance of $f$ to $\property$ is then defined as the minimum distance of $f$ to any function having the property, i.e. $\dist{f}{\property}\eqdef \inf_{g\in\property} \dist{f}{g}$ (equivalently, this is the minimum number of values one needs to change in $f$ so that it satisfies the property). We say that $f$ is \emph{\eps-far} from \property{} if $\dist{f}{\property} > \eps n$; otherwise, it is \emph{\eps-close}.

We work in the standard setting of property testing: Namely, an \emph{$\eps$-test} (for a fixed property $\property$) is a randomized algorithm which, given parameter $\eps\in(0,1]$ and black-box query access to an unknown function $f\colon[n]\to\R$, must, with probability $2/3$, accept if $f\in\property$ and reject if $f$ is \eps-far from \property. Moreover, for most of our results we shall be focusing specifically on \emph{one-sided} tests, i.e. those which are required to accept with probability one when $f\in\property$, and are only allowed to err when $f$ is \eps-far from \property. The reason for this focus on one-sided tests instead of the weaker two-sided requirement is motivated by the connection to \emph{finding} a violation to the property: Indeed, one-sided testing of $\pi$-freeness is essentially the algorithmic problem of finding a $\pi$-copy in a sequence that contains many pairwise-disjoint $\pi$-copies (instead of the simpler decision problem, which corresponds to the usual two-sided requirement).

With this is mind, we follow the standard use and define the \emph{query complexity} of an $\eps$-test $\Tester$ for $\pi$-freeness as the maximum number of black-box queries of the unknown sequence $f \colon [n] \to \R$ that $\Tester$ may make (as a function of the relevant parameters $n$, $\pi$, and $\eps$). The (one-sided) query complexity of $\pi$-freeness testing is then the minimum query complexity among all (one-sided) tests for it; our goal is to design optimal tests for $\pi$-freeness, i.e. ones that achieve optimal query complexity.

Throughout the paper, the parameter $n$ is always used to denote the length of the input sequence $f \colon [n] \to \R$, the parameter $k$ denotes the length of the forbidden order pattern $\pi$, and the proximity parameter of the test is denoted by $\eps$.
We use the notations $\tildeO{h},\tildeOmega{h}$ to hide polylogarithmic dependencies on $n$, i.e. for expressions of the form $\bigO{h \log^c n}$ and $\bigOmega{h \log^c n}$ (for some absolute constant $c$). We also write  $O_h$ (or, analogously, $\tilde{O}_h$) to denote expressions where the hidden multiplicative constant is allowed to depend (usually polynomially) on the additional parameter $h$.

\paragraph{Running time of our algorithms}
  All of our tests (except maybe that of~\autoref{sec:hierarchy}) run in time \emph{linear} in the number of queries they make. This is because they work by checking whether the queried subsequence contains the forbidden pattern. But this in turn can be performed efficiently, building on an algorithm of Guillemot and Marx~\cite{GuillemotMarx2014}  which determines whether a given sequence contains a fixed permutation in time that is linear in the size of the sequence. 

\paragraph{Distance function}
All results are stated here for the Hamming distance function, but they also hold for the stronger \emph{deletion distance}, defined as follows: $\text{dist}_{\text{del}}(f, g)$ is the minimal number of value modifications, deletions, and insertions needed to turn $f$ into $g$.
This follows from the fact that the Hamming distance and the deletion distance of a sequence $f$ to $\pi$-freeness are always equal: Indeed, if $S$ is a set of entries of a function $f \colon [n] \to \mathbb{R}$ whose deletion turns $f$ into a $\pi$-free sequence, then it is possible to turn $f$ into a $\pi$-free sequence using $|S|$ value modifications by initializing $T = S$ and iteratively applying the following until $T$ is empty: Find an $x \in T$ with a neighboring entry $y \notin T$, set $f(x) = f(y)$, and remove $x$ from $T$. This way, if $f$ restricted to $[n] \setminus S$ is $\pi$-free, then so is $f$ after these value modifications.

In particular, this implies that the distance of a sequence $f$ to $\pi$-freeness is closely related to the maximum size of a set $\mathcal{C}$ of pairwise-disjoint $\pi$-copies in $f$: On one hand, if $f$ is $\epsilon$-far from $\pi$-freeness then we cannot delete all $\pi$-copies in $\mathcal{C}$ with less than $\epsilon n$ entry deletions, so $|\mathcal{C}| \geq \epsilon n / |\pi|$. On the other hand, if $|\mathcal{C}| \geq \epsilon n$ then trivially $f$ is $\epsilon$-far from $\pi$-freeness.

\paragraph{Real-valued versus integer range}
As mentioned before, we are concerned here with real-valued functions, i.e. of the form $f\colon [n]\to\R$. However, this range is chosen merely for convenience and generality; all of our results should still hold when considering integer-valued functions, that is $f\colon [n]\to\N$.

\paragraph{On the respective parameters} As mentioned earlier, $\pi$-freeness is a property massively parameterized by a given permutation $\pi$. We regard the length of $\pi$ -- denoted $k$ -- as a constant, and generally focus on obtaining query complexities that are optimal up to a multiplicative factor that depends on $k$. Similarly, the proximity parameter $\eps$ is to be thought of as either a small constant or a function of the main parameter $n$ (the size of the domain), that slowly tends to zero.

%% file: sec-previous.tex
We describe here the previous state of knowledge on testing pattern-freeness; all results here are established in~\cite{NewmanRabinovich2017}, and focus on one-sided tests. 
\subsection{Sample-based non-adaptive upper bound}
\label{subsec:Newman_sample}
Any permutation of length $k$ has a non-adaptive one-sided test making $O(\eps^{-1/k} n^{1-1/k})$ queries. This is the sample-based test, that samples a uniformly random set of elements in the input sequence of the required size and accepts the input (i.e., indicates that it is $\pi$-free) if the queried subsequence is $\pi$-free. In what follows, this test is called the \emph{sampler}.

\subsection{Efficient non-adaptive testing of monotone permutations}
\label{subsec:Newman_monotone}
$\pi$-freeness is efficiently testable non-adaptively when $\pi$ is monotone: 
For any $k > 1$, if $\pi = (1,2,\ldots,k)$ or $\pi = (k,k-1,\ldots,1)$ then $\pi$-freeness has a one-sided non-adaptive $\eps$-test making $(\eps^{-1} \log{n})^{O(k^2)}$ queries.
The queries are chosen (non-adaptively) using a dedicated algorithm called the \textsf{dyadic sampler}, that iteratively tries to ``guess'' the typical structure of a $\pi$-copy and query according to this guess. 
This settles the special case of monotone permutations, up to a factor polynomial in $\eps^{-1} \log{n}$.

\subsection{Permutations of length $3$, and an exponential gap}
\label{subsec:Newman_3}
Due to symmetry considerations, to understand the behavior of non-monotone permutations of length $3$ it is enough to consider the permutation $\pi = (1, 3, 2)$. 
For this choice of $\pi$, it is shown that:
\begin{itemize}
	\item There is an \emph{adaptive} one-sided $\eps$-test for $\pi$-freeness making $(\eps^{-1} \log{n})^{O(1)}$ queries.
	\item Any \emph{non-adaptive} $1/9$-test for $\pi$-freeness has query complexity $\Omega(\sqrt{n})$ -- an exponential separation from the adaptive case! It is interesting to note that while the lower bound in~\cite{NewmanRabinovich2017} was only obtained for one-sided tests, a similar lower bound for two-sided tests may be derived using similar (yet more technical) ideas.
	\item There is a \emph{non-adaptive} one-sided $\eps$-test for $\pi$-freeness making $\sqrt{n} (\eps^{-1}\log{n})^{O(1)}$ queries. Thus, the non-adaptive bounds for $\pi = (1,3,2)$ are tight up to an $(\eps^{-1} \log{n})^{O(1)}$ factor.
\end{itemize}

\subsection{Non-adaptive lower bounds, and separations between permutations}
\label{subsec:Newman_lower}
The $\Omega(\sqrt{n})$ non-adaptive lower bound from~\autoref{subsec:Newman_3} actually applies to \emph{any} non-monotone permutation.
Moreover, this bound can be strengthened for certain permutations: For any odd $k$, any one-sided non-adaptive test for the permutation $\pi = (1, k, k-1, 2, 3, k-2, \ldots, (k+1)/2)$ requires $\Omega(n^{1 - 2/(k+1)})$ queries.

\subsection{Discussion on previous results}
\label{subsec:discussion_previous}
The results in~\cite{NewmanRabinovich2017} essentially settle two special cases: The monotone permutations of any length, and the permutations of length $3$.
However, the general task of understanding the query complexity of optimal tests for $\pi$-freeness -- for any $\pi$ -- both in the adaptive and the non-adaptive case, has remained wide open.
The major open problems that Newman et al.\@ proposed are the following.
\begin{description}
	\item[Adaptive case] Is it true that $\pi$-freeness is testable adaptively with query complexity polylogarithmic in $n$ for \emph{any} permutation $\pi$?
	\item[Non-adaptive case] How does the structure of a permutation $\pi$ correlate with the query complexity of an optimal (one-sided) non-adaptive test for $\pi$-freeness? 
	In particular, do there exist infinitely many permutations $\pi$ for which $\pi$-freeness is testable with query complexity that is $O(n^{0.99})$?
\end{description}

%% file: sec-contributions.tex
In this paper, we mainly address the non-adaptive case, achieving good (though not yet complete) understanding of this case. Along the way, we discover many interesting and surprising phenomena. The details are presented in Subsections~\ref{subsec:tight_upper_bound} to~\ref{subsec:hierarchy_non_adaptive}.

Additionally, we explore how partial adaptivity helps in the problem of testing $\pi$-freeness, in particular for the permutation $\pi = (1, 3, 2)$. We observe that a hierarchy of adaptivity exists for this problem, making it the first known natural example of a single property with such an hierarchy. The main results concerning partial adaptivity are presented in~\autoref{subsec:hierarchy_adaptivity_132}.

\subsection{Tight non-adaptive upper bound}
\label{subsec:tight_upper_bound}
Our first main result is an improved upper bound for the non-adaptive one-sided case.
\begin{theorem}
	\label{thm:main_thm_upper1}
	For any permutation $\pi$ of length $k \geq 3$, $\pi$-freeness has a one-sided non-adaptive $\eps$-test whose query complexity is $O(\eps^{-\frac{1}{k-1}}n^{1 - \frac{1}{k-1}})$.
\end{theorem}
This bound improves upon all previously known upper bounds for non-monotone permutations, as the query complexity it suggests is better than both the sample-based upper bound from~\autoref{subsec:Newman_sample} and the upper bound for permutations of length $3$, from~\autoref{subsec:Newman_3}.

At first glance, an upper bound of $O(\eps^{-\frac{1}{k-1}}n^{1 - \frac{1}{k-1}})$ seems 
to only be a slight improvement over the $O(\eps^{-\frac{1}{k}}n^{1 - \frac{1}{k}})$ sample-based upper bound.
However, quite surprisingly, this upper bound is tight in both $n$ and $\eps$ for any $k \geq 3$. 
In other words, the optimal non-adaptive one-sided test for some permutations is only slightly more query-efficient than the sampler!
\begin{theorem}
	\label{thm:upper_bound_is_tight}
	Let $\pi$ be a permutation of length $k \geq 3$, and suppose that $|\pi^{-1}(1) - \pi^{-1}(k)| = 1$. Then the query complexity of any non-adaptive one-sided $\eps$-test for $\pi$-freeness is $\Omega\left(\eps^{-\frac{1}{{k-1}}}n^{1 - \frac{1}{k-1}}\right)$.
\end{theorem}
This improves the non-adaptive lower bounds for any permutation of this type, whose length is at least four. For the non-monotone permutations of length $3$, this results determines the correct dependence in $\eps$.

The combination of~\autoref{thm:upper_bound_is_tight} with the results in~\autoref{subsec:Newman_monotone} demonstrates a surprising phenomenon:
While the deletion distance between the permutations $\pi_1 = (1, 2, \ldots, k)$ and $\pi_2 = (k, 1, 2, \ldots, k-1)$ is only $2$, the query complexity of non-adaptive one-sided testing for $\pi_1$-freeness differs significantly from that of $\pi_2$-freeness. For $\pi_1$-freeness this query complexity is polylogarithmic in $n$, and so $\pi_1$ is the easiest to test among permutations of length $k$, while $\pi_2$-freeness has a query complexity of $\Theta\left(\eps^{-\frac{1}{k-1}}n^{1-{\frac{1}{k-1}}}\right)$, making $\pi_2$ one of those permutations that are hardest to test with non-adaptive one-sided tests.

\noindent The proofs of Theorems~\ref{thm:main_thm_upper1} and~\ref{thm:upper_bound_is_tight} appear in Sections~\ref{sec:upper} and~\ref{sec:lower}, respectively.

\subsection{Almost tight non-adaptive bounds for random permutations}
The next lower bound is perhaps even more surprising. It provides (along with~\autoref{thm:main_thm_upper1}) an almost tight bound on the query complexity of an optimal non-adaptive one-sided test for almost all permutations, implying that this query complexity is usually only marginally better than that of the sampler.
\begin{theorem}
	\label{thm:lower_bound_random_permutation}
	Let $\pi$ be picked uniformly at random from all permutations of length $k$. The following holds with probability $1-o(1)$ (where the $o(1)$ term tends to zero as $k \to \infty$):
	The query complexity of any non-adaptive one-sided $\eps$-test for $\pi$-freeness is $\Omega\left(\eps^{-\frac{1}{k-3}} n^{1 - \frac{1}{k-3}}\right)$.
\end{theorem}


\subsection{General permutation-dependent non-adaptive lower bound}
\label{subsec:general_lower_bound}
Both Theorems~\ref{thm:upper_bound_is_tight} and~\ref{thm:lower_bound_random_permutation}
are actually special cases of a general permutation-dependent lower bound that we establish. This lower bound applies to any permutation, 
and depends heavily on the structure of the permutation. We believe that this lower bound is tight (up to polylogarithmic factors) for any permutation.

Interestingly, it is not clear how to describe the lower bound in a compact closed form, but given a permutation $\pi$ of length $k$, the corresponding bound can be computed in constant time (that depends only on $k$).
Later, as an important special case of this strong yet hard-to-digest bound, we provide a slightly weaker permutation-dependent lower bound that has a more natural combinatorial characterization, and is therefore easier to analyze. See~\autoref{thm:entangling_to_unique} and the resulting Corollaries~\ref{cor:entangling_lower_bound} and~\ref{cor:max_diff_lower_bound} for more details. 

\noindent In order to describe our general lower bound, we shall first provide some definitions.
\begin{definition}
	\label{def:USPN}
	Let $\pi = (\pi_1, \ldots, \pi_k)$ be a permutation of length $k$. 
	A subsequence $\sigma$ of $\pi$ is \emph{consecutive} if $\sigma = (\pi_i, \ldots, \pi_j)$ for some $1 \leq i \leq j \leq k$; in this case we write $\sigma = \pi[i,j]$.
	
	A \emph{partition} $\Lambda = (\sigma_1, \ldots, \sigma_\ell)$ of the permutation $\pi$ consists 
	of consecutive subsequences $\sigma_1 = \pi[1,r_1], \sigma_2 = \pi[r_1+1, r_2], \ldots, \sigma_\ell = \pi[r_{\ell-1}+1, k]$, and its \emph{size} is $|\Lambda| = \ell$.
	
	A \emph{signed partition} $P = (\Lambda, S)$ of the permutation $\pi$ consists of a partition $\Lambda$ as above, and a sign vector $S = (s_1, \ldots, s_\ell) \in \{+,-\}^{\ell}$. 
	For any $\sigma_i$ of length bigger than one, the corresponding sign $s_i$ must satisfy the following. If $\min \sigma_i$ appears before $\max \sigma_i$ in $\pi$, then the direction sign of $\sigma_i$ is $-$, and otherwise, the direction sign is $+$. The \emph{size} of $P$ is $|\Lambda| = |S| = \ell$.
	
	Let $P$ be a signed partition as above.
	Define $r_0 = 0$, and for any $1 \leq i \leq \ell$, denote the length of $\sigma_i$ by $k_i$ (so $\sum_{i=1}^{\ell} k_i = k$).
	Consider the sequence $f_P\colon [k^2] \to \mathbb{R}$ defined as follows. 
	For any $1 \leq j \leq k_i$ and $0 \leq m \leq k-1$, we take $f_P(r_i k + m k_i + j) = m + \pi_{r_i + j} / 2k$ for any $0 \leq i \leq \ell-1$ where $s_i$ is $+$, and $f_P(r_i k + m k_i + j) = (k - 1 - m) + \pi_{r_i + j} / 2k$ for any $\ \leq i \leq \ell-1$ where $s_i$ is a $-$. 
	
	Note that for any $0 \leq m \leq k-1$, the set of all entries $x \in [k^2]$ satisfying $m < f_P(x) < m+1$ is a $\pi$-copy. We say that such a $\pi$-copy is \emph{trivial}.
	We say that $P$ is \emph{unique} if $f_P$ does not contain non-trivial $\pi$-copies, and denote by $\mathcal{U}(\pi)$ the set of all unique signed partitions of $\pi$.
	Finally, the \emph{unique signed partition number} (USPN) of $\pi$ is $u(\pi) = \max_{P \in \mathcal{U}(\pi)} |P|$.
\end{definition}
\noindent Our lower bound for testing $\pi$-freeness is closely related to the USPN of $\pi$.
\begin{theorem}
	\label{thm:main_thm_lower2}
	Let $\pi$ be any permutation. Any non-adaptive one-sided $\eps$-test for $\pi$-freeness has query complexity $\Omega\left(\eps^{-1/u(\pi)} n^{1-1/u(\pi)}\right)$.
\end{theorem}
The USPN of a permutation obviously depends only on the permutation (and not on the input sequence size), so it can be computed in constant time, that depends only on $k$. Thus, given a permutation $\pi$ and parameters $n, \eps$, one can compute the lower bound obtained from~\autoref{thm:main_thm_lower2} in constant time.

The proof of~\autoref{thm:upper_bound_is_tight} follows from~\autoref{thm:main_thm_lower2}
by showing that for any permutation $\pi$ of length $k$ which satisfies $|\pi^{-1}(1) - \pi^{-1}(k)| = 1$, it holds that $u(\pi) = k-1$; actually these are the only permutations of length $k$ whose USPN is $k-1$, and no permutation of length $k > 1$ has USPN that equals $k$, as can be derived from results that are discussed later.

We conjecture that the lower bound of~\autoref{thm:main_thm_lower2} is tight up to a multiplicative term that is polynomial in $\eps$ and $\log{n}$.
That is, we conjecture that the USPN, $u(\pi)$, is the correct parameter of $\pi$ that determines how hard it is to non-adaptively test $\pi$-freeness using one-sided tests. 

\begin{conjecture}
	\label{open:non_adaptive}
	For any permutation $\pi$ of any length, $\pi$-freeness has a non-adaptive one-sided $\eps$-test making $\tilde{\Theta}_{\eps} \left( n^{1-1/u(\pi)} \right)$ queries.
\end{conjecture}
A multiplicative term of $\log n$ is necessary to make~\autoref{open:non_adaptive} hold for monotone permutations $\pi$ (for which $u(\pi) = 1$), since there is a lower bound of $\Omega(\log n)$ for testing monotonicity~\cite{Fischer2004}, that can be generalized to testing $\pi$-freeness for any permutation $\pi$ of length at least $2$. 

For non-monotone permutations, 
an even stronger conjecture can be given, namely that the number of queries required by a non-adaptive one-sided $\eps$-test is $\Theta_{\eps}(n^{1 - 1/u(\pi)})$ (without the polylogarithmic term in $n$).

\subsection{Combinatorial characterizations related to the general lower bound}
\label{subsec:character_lower_bound}
Motivated by~\autoref{thm:main_thm_lower2}, it is desirable to find natural necessary and sufficient combinatorial conditions for uniqueness of a signed partition of a given permutation $\pi$.
Our next main result provides a useful sufficient condition. For the result, we need some more
definitions. 
\begin{definition}
 Let $\sigma = \pi[x,y]$ and $\sigma' = \pi[x',y']$ of $\pi$ be disjoint consecutive subsequences of length at least two, and let $x' \leq m, M \leq y'$ be the indices satisfying $\pi_m = \min \pi[x', y']$ and $\pi_M = \max \pi[x', y']$. 
We say that $\sigma'$ is \emph{shadowed} with respect to $\sigma$ if one of the following holds.
	\begin{itemize}
		\item $x' > y$, $m < M$, and $\pi_{x'-1} > \pi_M$.
		\item $x' > y$, $m > M$, and $\pi_{x'-1} < \pi_m$.
		\item $y' < x$, $m < M$, and $\pi_{y'+1} < \pi_m$.
		\item $y' < x$, $m > M$, and $\pi_{y'+1} > \pi_M$.
	\end{itemize} 
	
	An \emph{entangling} of $\pi$ is a collection $E = (\sigma_1, \ldots, \sigma_t)$ of pairwise disjoint consecutive subsequences of $\pi$, where $\sigma_i = \pi[a_i,b_i]$ for any $1 \leq i \leq t$, satisfying the following.
	\begin{itemize}
		\item For any $2 \leq j \leq t$, the following holds. Either $a_j > b_1$ and $\min_{i < j} \min{\sigma_i} < \pi_{a_j} < \max_{i < j} \max {\sigma_i}$, or $b_j < a_1$ and $\min_{i < j} \min{\sigma_i} < \pi_{b_j} < \max_{i < j} \max {\sigma_i}$.
		\item For any $2 \leq j \leq t$, $\sigma_j$ is not shadowed with respect to $\sigma_1$.
		\item For any $1 \leq \ell \leq k$, there exists $\sigma \in E$ such that $\min{\sigma} \leq \ell \leq \max{\sigma}$.
	\end{itemize}
	For the above entangling $E$ of $\pi$, define $\Lambda(E)$ as the partition of $\pi$ in which $\sigma_1, \ldots, \sigma_t$ serve as parts, and any element of $\pi$ not in $\bigcup_{i=1}^{t} \sigma_i$ has its own part.
	Denote $d(E) = |\Lambda(E)| = k - \sum_{\sigma \in E} (|\sigma|-1)$. Finally, the \emph{entangling number} of $\pi$ is $d(\pi) = \max_{E} \{ d(E) \}$ where $E$ ranges over all
	valid entanglings of $\pi$.
\end{definition}

\begin{theorem}
	\label{thm:entangling_to_unique}
	For any permutation $\pi$ and entangling $E$ of $\pi$, there exists $S \in \{+,-\}^{|E|}$ for which the signed partition $P = (\Lambda(E), S)$ is unique.
	In particular, $d(\pi) \leq u(\pi)$ for any permutation $\pi$.
\end{theorem}

The following is an immediate yet important corollary of Theorems~\ref{thm:main_thm_lower2} and~\ref{thm:entangling_to_unique}.
\begin{corollary}
	\label{cor:entangling_lower_bound}
	For any permutation $\pi$, any non-adaptive one-sided $\eps$-test for $\pi$-freeness must make $\Omega\left(\eps^{-1/d(\pi)} n^{1-1/d(\pi)}\right)$ queries.
\end{corollary} 
A useful simple special case of~\autoref{cor:entangling_lower_bound} is the following.
\begin{corollary}
	\label{cor:max_diff_lower_bound}
For a permutation $\pi = (\pi_1, \ldots, \pi_k)$, let $m(\pi) = \max_{1\leq i \leq k-1} |\pi_{i+1} - \pi_{i}|$. Then $m(\pi) \leq d(\pi)$, and in particular, any non-adaptive one-sided $\eps$-test for $\pi$-freeness must make $\Omega\left(\eps^{-1/m(\pi)} n^{1-1/m(\pi)}\right)$ queries.
\end{corollary} 

Note that a permutation $\pi$ of length $k$ with $|\pi^{-1}(1) - \pi^{-1}(k)| = 1$  satisfies $m(\pi) = k-1$, so~\autoref{thm:upper_bound_is_tight} is actually a special case of~\autoref{cor:max_diff_lower_bound}.

~\autoref{thm:lower_bound_random_permutation} follows from~\autoref{cor:entangling_lower_bound} by observing that $d(\pi) \geq k-3$ holds w.h.p.\@ for a random permutation $\pi$ of length $k$; actually, both $d(\pi)$ and $u(\pi)$ are concentrated in the values $k-2$ and $k-3$, as $u(\pi) = k-1$ holds with probability $O(1/k)$.

There exist permutations $\pi$ for which $d(\pi) < u(\pi)$. In particular, partitions with a unique signed form are not necessarily entanglings, so the sufficient condition for uniqueness from~\autoref{thm:entangling_to_unique} is not a necessary one. For example, one can verify that $\pi = (4, 1, 2, 5, 6, 3)$ satisfies $d(\pi) = 3$ but $u(\pi) = 4$; a unique signed partition of size $4$ for $\pi$ is $(\Lambda, S)$ where $\Lambda = ((4, 1), 2, 5, (6, 3))$ and $S = (+, -, -, +)$). 

The following necessary condition for being a unique signed partition is easy to prove.
\begin{observation}
	\label{obs:neccesary_condition_unique}
	Let $\pi$ be a permutation of length $k$ and let $P = (\Lambda, S)$ be a unique signed partition for $\pi$. Then $\Lambda$ satisfies the following conditions.
	\begin{itemize}
		\item For any $1 \leq \ell \leq k$ there exists $\sigma \in \Lambda$ of length bigger than one, such that $\min{\sigma} \leq \ell \leq \max{\sigma}$.
		\item Let $\sigma \in \Lambda$ with $|\sigma| > 1$. If $\max \sigma < k$ then there exists $\sigma' \in \Lambda$ satisfying $\min \sigma' < \max \sigma < \max \sigma'$, and similarly,
		if $\min \sigma > 1$ then there exists $\sigma' \in \Lambda$ satisfying $\min \sigma' < \min \sigma < \max \sigma'$.
	\end{itemize}
\end{observation}

The size $|\Lambda|$ of the largest partition $\Lambda$ of $\pi$ satisfying the conditions in~\autoref{obs:neccesary_condition_unique} might be bigger than the USPN of $\pi$. For example, the partition $\Lambda = ((5, 1), 3, 2, 7, 6, (8, 4))$ of the permutation $\pi = (5, 1, 3, 2, 7, 6, 8, 4)$ satisfies these conditions, but one can verify that it is not a unique signed partition. By~\autoref{obs:neccesary_condition_unique}, none of the other partitions of $\pi$ of size $6$ have a unique signed form, so $u(\pi) < 6 = |\Lambda|$. In fact, $u(\pi) = 5$ in this case, as $((1,3), (2,7), (6,8))$ is an entangling.

\subsection{Permutation-dependent hierarchy in the non-adaptive case}
\label{subsec:hierarchy_non_adaptive}
The statement of~\autoref{open:non_adaptive} suggests that there is a permutation-dependent hierarchical behavior of the query complexity for one-sided non-adaptive testing of $\pi$-freeness as a function of $\pi$. 
The following result verifies that such an hierarchical structure indeed exists.
\begin{theorem}
	\label{thm:hierarchy_nonadaptive}
	For any two positive integers $k \geq 2$ and $1 \leq \ell \leq k-1$, there is a permutation $\pi$ of length $k$ with $m(\pi) = \ell$, for which the optimal non-adaptive $\eps$-test makes $\tilde{\Theta}_{\eps} \left(n^{1-1/\ell} \right)$ queries, where the $\tilde{\Theta}_\eps$ notation hides a term polynomial in $\log{n}$ and $\eps$.
\end{theorem}
In particular, we conclude that for any positive integer $\ell$, there exist infinitely many permutations $\pi$ for which the query complexity of one-sided non-adaptive testing of $\pi$-freeness is $\tilde{\Theta}_{\eps}(n^{1 - 1/\ell})$.
This answers and generalizes the open question of Newman et al.\@~\cite{NewmanRabinovich2017}, who asked whether there exist infinitely many permutations $\pi$ 
that have a non-adaptive one-sided test for $\pi$-freeness making at most $O(n^{0.99})$ queries (for a fixed \eps).  

\subsection{Hierarchy of Adaptivity}
\label{subsec:hierarchy_adaptivity_132} 
Recent work of Canonne and Gur~\cite{CG:17} introduced the notion of \emph{amount of adaptivity} in property testing, which they define as follows.\footnote{Actually, for clarity of exposition we slightly depart from the notation of~\cite{CG:17}, and use \emph{$r$-round test} for what they refer to as \emph{$(r-1)$-round-adaptive test}. The reason is that, with our convention, an $r$-round test is an algorithm which proceeds in $r$ rounds; while in theirs, an {$r$-round test} is an algorithm which proceeds in $r+1$ rounds, with $r$ adaptive ones and one (the very first) being non-adaptive.}{} An \emph{$r$-round test} for some property $\property$ is an algorithm which proceeds in $r$ stages. At each stage, it produces and makes a batch of queries to the function, which cannot depend adaptively on each other (i.e., these queries are among themselves non-adaptive), and receives all answers to these queries. It then produces the queries for the next stage, which can depend adaptively on the answers just received. At the end of the $r$ stages (the first one being fully non-adaptive, and each query in the subsequent rounds depending adaptively on those made in the previous rounds only), the algorithm must accept or reject the function; the query complexity is then the total number of queries made overall. 

Note that in this formalization, non-adaptive tests correspond to $1$-round tests, while adaptive ones are those with unbounded number of rounds. In~\cite{CG:17}, the authors establish a strong hierarchy theorem, showing separations between $r$-round tests and $(r+1)$-round tests, for any integer $r$ -- albeit for a somewhat contrived family of properties. They also establish another such theorem, slightly weaker, but this time for a \emph{natural} property (of graphs). However, both results have the issue of producing a different property for every $r$: that is, ``for every integer $r$ there exists a property $\property_r$ hard to test in $r$ rounds, but easy in $r+1$.'' 
Determining whether there exists a \emph{single} natural property which would witness such a hierarchy -- ``there exists a property $\property$ which is, for infinitely many integers $r<r'$,  harder to test in $r$ rounds than in $r'$'' -- is posed as an open question in~\cite{CG:17}.

In this work, we give a positive answer for this question by analyzing tests for the specific pattern $\pi=(1,3,2)$. Newman et al. established in~\cite{NewmanRabinovich2017} an exponential gap between the query complexity of \emph{adaptive} and \emph{non-adaptive} tests; focusing on $r$-round tests, and building on their algorithm, we establish a finer separation for testing this particular pattern, leading to our adaptivity hierarchy result.
\begin{restatable}{theorem}{theohierarchub}\label{theo:hierarchy:ub}
    For every $1\leq r\leq (\log n)^{O(1)}$, there exists an $r$-round (one-sided) testing algorithm for $(1,3,2)$-freeness with query complexity $\tildeO{\eps^{-1} n^{\frac{1}{r+1}}}$. 
\end{restatable}
\noindent We then complement the upper bound part of our adaptivity separation by the following lower bound counterpart:
\begin{restatable}{theorem}{theohierarchlb}\label{theo:hierarchy:lb}
    For every $0\leq r\leq (\log \log n)^{O(1)}$, any $r$-round (\emph{two}-sided) testing algorithm for $(1,3,2)$-freeness must make $\bigOmega{n^{1/2^{r+3}}}$ queries.
\end{restatable}
Note that the type of tests involved in Theorems \ref{theo:hierarchy:ub} and \ref{theo:hierarchy:lb} is best possible.
Indeed, it implies that there is a test $\Tester$ using $2^{r+3}$ rounds whose number of queries is polynomially better than that of any one-sided or two-sided $r$-round test, and in addition that this $\Tester$ belongs to the more restricted class of one-sided tests.
We remark that one could hope for an even stronger theorem, which would separate $r$ rounds from $r+1$ (as opposed to $r$ vs. $\exp(r)$ as above). We conjecture that the best $r$-round-adaptive test has query complexity $\tilde{\Theta}_{\epsilon}\left(n^{1 / (2r + O(1))}\right)$.

%% file: sec-discussion.tex
The problem of (one-sided) testing of $\pi$-freeness demonstrates a wide array of interesting phenomena: An exponential separation between the adaptive and the non-adaptive case, surprising hardness results and permutation-dependent hierarchical behaviors in the non-adaptive case, 
and a hierarchy of adaptivity that is the first of its kind. We believe that these results serve as a strong motivation to try to achieve a complete understanding of the problem.
Below we suggest several possible directions for future research.

\paragraph{The adaptive case}
Testing $\pi$-freeness in the adaptive case is still far from being understood. In particular, the question whether all permutations are testable adaptively with number of queries that is polylogarithmic in $n$ is still wide open, even if we allow for two-sided tests. At this point, this seems to be the most intriguing open question regarding testing $\pi$-freeness. 

\paragraph{Improving bounds in the non-adaptive case}
While our understanding of the non-adaptive case is far better than that of the adaptive case, 
there are still gaps in it. The main goal here is to obtain good permutation-dependent upper bounds:~\autoref{open:non_adaptive} states that our lower bound is actually tight, and it will be obviously interesting to understand if it holds. 

\paragraph{Understanding the USPN}
Another interesting direction would be to obtain a simple complete combinatorial characterization of the USPN of any given permutation. Currently we have lower and upper bounds for the USPN of a permutation (\autoref{thm:entangling_to_unique} and~\autoref{obs:neccesary_condition_unique}, respectively), that are usually tight for small permutations,
and we know that the USPN of a permutation is computable in constant time.

\paragraph{Two-sided testing}
All known results so far are for one-sided testing, aside from our two-sided lower bound in the partially adaptive setting. It is worth to note that the $\Omega_{\eps}(n^{1/2})$ lower bound 
for one-sided testing of all non-monotone permutations can be (carefully) translated into the same bound for two-sided tests. However, the proofs of other one-sided non-adaptive lower bounds do not seem to translate well to the two-sided setting. 

Therefore, it will be interesting to understand what is the query complexity of optimal \emph{two-sided} tests, both in the adaptive and the non-adaptive case. Specific questions of interest include (but are not limited to) the following: When do the non-adaptive two-sided lower bounds match the one-sided ones? Can one obtain a general two-sided upper bound that beats the tight one-sided upper bound of~\autoref{thm:main_thm_upper1} for permutations of size bigger than three? Does two-sidedness help testing in the adaptive case? 

\paragraph{Families of forbidden order patterns}
It will be interesting to investigate the case where more than one order pattern is forbidden (note that there are families for which the question does not make sense; for example, the famous theorem of Erd\H{o}s-Szekeres~\cite{ErdosSzekeres1935} implies that any sequence of length at least $k^2 - 2k + 2$ must contain one of the monotone permutations of length $k$).
As mentioned in~\cite{NewmanRabinovich2017}, all one-sided upper bounds from the single-pattern case carry over to the multiple-pattern case, but the lower bounds do not; for example, there exists a family of two non-monotone permutations of size $3$ that has a one-sided non-adaptive test whose query complexity is polylogarithmic in $n$. Some specific open questions of interest: Is the upper bound from~\autoref{thm:main_thm_upper1} tight in this case? How does the non-adaptive family-dependent hierarchy look like?

\paragraph{Forbidden order patterns in multi-dimensional structures}
How does $\pi$-freeness behave in structures of higher dimensions, such as the hypergrid or the Boolean hypercube?
The sample-based upper bound from~\autoref{subsec:Newman_sample} still holds in these cases, provided that the input contains many pairwise-disjoint copies of the forbidden structure $\pi$. However, in contrast to the one-dimensional case, it is then no longer clear whether being far from $\pi$-freeness implies that the input indeed has many pairwise-disjoint $\pi$-copies. Interestingly, recent work of Grigorescu, Kumar, and Wimmer~\cite{GKW:17} gives strong evidence that testing order pattern freeness on the hypercube is hard.

%% file: sec-upper.tex
In this section we provide the proof of~\autoref{thm:main_thm_upper1}.
The test that is used to prove the upper bound is one-sided, and indicates that the input sequence $f \colon [n] \to \R$ has a $\pi$-copy only if it finds one. Thus, the testing problem reduces to the following search problem: Given query access to un unknown sequence $f$ that is $\eps$-far from $\pi$-freeness, the goal is to find a $\pi$-copy in $f$. Here and henceforth, we omit floor and ceiling signs, as they do not make an essential different in the arguments.

The proof of~\autoref{thm:main_thm_upper1} follows immediately from the next lemma, which provides a sublinear algorithm to find a $\pi$-copy in a sequence $f$, assuming that $f$ is far enough from $\pi$-freeness. 
\begin{lemma}
	\label{lem:algorithm_upper_bound}
Let $\pi$ be a permutation of length $k \geq 3$, and suppose that $f \colon [n] \to \R$ is $\eps$-far from $\pi$-freeness for some $\eps \geq c_k n^{-1/9}$, where $c_k$ depends only on $k$, and $n$ is large enough (as a function of $k$). Then there is an algorithm that finds, with probability $2/3$, a copy of $\pi$ in $f$ by querying $O(\eps^{-\frac{1}{k-1}} n^{1 - \frac{1}{k-1}})$ entries in $f$.
\end{lemma}
\begin{remark}
	\autoref{lem:algorithm_upper_bound} is stronger than what is needed to obtain a one-sided test, in the sense that $\eps$ is allowed to scale with $n$; for the proof of~\autoref{thm:main_thm_upper1} a lemma that applies to a constant $\eps$ would have been sufficient. However, the added flexibility of the lemma reflects that the statement of~\autoref{thm:main_thm_upper1} would still be true should we take $\eps^{-1}$ as a slowly-growing function of $n$.
\end{remark}
\begin{proof}
	The proof idea is as follows. Given an input sequence $f \colon [n] \to \R$, we partition $[n]$ into a collection $\mathcal{I}$ of $n/m$ intervals of size $m$ each, for a suitable choice of $m$; we may assume, for convenience, that $m$ divides $n$. 
	Suppose that $f$ is $\eps$-far from $\pi$-freeness. Then $f$ contains a set $\mathcal{A}$ of $\eps n / k$ pairwise disjoint $\pi$-copies. We consider two cases, where for each of the cases the queries made are different. Our actual algorithm is a combination of the algorithms for each of the cases.
	
	In the first case, 
	most $\pi$-copies in $\mathcal{A}$ have at least two entries in the same interval; the algorithm for this case queries a set of whole intervals, chosen uniformly at random, and a set of single elements, also chosen uniformly at random. The analysis of this case does not use the fact that $\pi$ is a permutation. 
	In the second case, most $\pi$-copies in $\mathcal{A}$ do not have two entries in the same interval, and it can be shown that the sampler (which samples entries of $f$ uniformly at random) suffices for this case. Here we do use the fact that $\pi$ is a permutation, and the analysis actually shows that the required number of queries is much smaller (for constant $\eps$) than in the first case. 
	
	We now give the full details.
	Pick the interval size to be $m = (\eps n)^{1-1/(k-1)}$, and write $\pi = (\pi_1, \ldots, \pi_k)$.
	The $\pi$-copies are represented in $\mathcal{A}$ as $k$-tuples $t = (t_1, \ldots, t_k)$
	where $t_i$ is the location of the element corresponding to $\pi_i$
	in the copy.
	Write $\mathcal{A} = \mathcal{B} \cup \mathcal{C}$, where $\mathcal{B}$ contains all $\pi$-copies from $\mathcal{A}$ that have at least two entries in the same interval, and $\mathcal{C}$ contains all $\pi$-copies that have at most one entry in each interval. Then either $|\mathcal{B}| \geq \eps n / 2k$ or $|\mathcal{C}| \geq \eps n / 2k$.
	
	\paragraph{Case 1: $|\mathcal{B}| \geq \eps n / 2k$}
	Our algorithm for this case is described as follows. We first pick a set $Q_1 \subseteq \mathcal{I}$ of intervals, where every $I \in \mathcal{I}$ is included in $Q_1$ with probability $p = c m / \eps n= c(\eps n)^{-1/(k-1)}$, independently of other intervals. Here $c = 100k^2$ is a constant that depends (polynomially) on $k$.  
	Next, we pick a set $Q_2$ of elements from $[n]$, where each element is added to $Q_2$ with probability $p$, uniformly and independently of other elements. Up to this point, the algorithm does not make any queries. 
	
	\subparagraph{An independent sampling trick}
	Variants of the following simple idea are used several times along the paper.
	Let $E_{\text{found}}$ be the event that
	the subsequence of $f$ induced by $Q_1$ and $Q_2$ contains a $\pi$-copy.
	Let $E_{\text{big}}$ be the event that $|Q_1| > 100 c / \eps$ or $|Q_2| > 100 c m / \eps$. 
	By Markov's inequality, $\Pr(E_{\text{big}}) \leq 1/50$. If $E_{\text{big}}$ occurs, then the algorithm stops without making any queries (and hence it does not find a $\pi$-copy in $f$).
	If $E_{\text{big}}$ has not occurred, then the algorithm now queries all elements induced by $Q_1$ and $Q_2$.
	Thus, the algorithm finds a $\pi$-copy if and only if $E_{\text{found}}$ occurs and $E_{\text{big}}$ does not occur.
	The number of queries made is at most $200cm/\eps = O(\eps^{-\frac{1}{k-1}} n^{1 - \frac{1}{k-1}})$, as desired. The probability that the algorithm finds a $\pi$-copy
	is at least $\Pr(E_{\text{found}}) - \Pr(E_{\text{big}}) \geq \Pr(E_{\text{found}}) - 1/50$. Thus, it remains to show that $\Pr(E_{\text{found}}) \geq 2/3 + 1/50$ (note that we consider here the unconditional probability of $E_{\text{found}}$, and in particular, we do not condition on $E_{\text{big}}$ not happening). 
	
	\subparagraph{Analyzing the event $E_{\text{found}}$}
	For each $I \in \mathcal{I}$, let $t_I$ denote the number of $\pi$-copies from $\mathcal{B}$ that have at least two entries in $I$, and let $t = \sum_{I \in \mathcal{I}} t_I$, so $ \eps n / 2k \leq t \leq \eps n$.
	let $X$ be the random variable that counts the number of $\pi$-copies from $\mathcal{B}$ that have at least two entries in some $I \in Q_1$. The expectation of $X$ is $\mathbb{E}[X] = t p \geq cm/2k = 50km$, and the variance of $X$ is bounded by
	$$ \mathbb{E}[X^2] \leq p \sum_{I \in \mathcal{I}} t_I^2 \leq p m^2 \eps n / m = cm^2
	$$
	where the second inequality follows from convexity arguments, using the facts that $0 \leq t_I \leq m$ for every $I$ and $\sum t_I = t \leq \eps n$. Thus, the standard deviation of $X$ is bounded by $m\sqrt{c} = 10km$. By Chebyshev's inequality, we get that $X \geq m$ with probability at least $9/10$. 
	
	Assume now that $X \geq m$, that is, there exists a set $\mathcal{B}' \subseteq \mathcal{B}$ of $m$ $\pi$-copies that have at least two of their entries in intervals from $Q_1$. For each such copy, the event that all other $k-2$ (or less) entries of it are in $Q_2$ has probability at least $p^{k-2} = c^{k-2} / m$, and is independent of the corresponding events of the other copies in $\mathcal{B}'$. 
	Thus, the probability that none of these events occurs is bounded by $(1-c^{k-2}/m)^m \leq e^{-c^{k-2}} < 1/100$.
	This finishes the proof.
	
	\paragraph{Case 2: $|\mathcal{C}| \geq \eps n / 2k$}
	We start with some notation. For a copy $t = (t_1, \ldots, t_k) \in \mathcal{C}$, we define $I_i(t)$ as the interval in $\mathcal{I}$ containing $t_i$. 

	\subparagraph{Non-extremal $\pi$-copies}
	For any interval $I \in \mathcal{I}$, let $y_1 \leq \ldots \leq y_m$ be the elements of
	$f(I) = \{f(x) : x \in I \}$, and let $y_I^- = y_{\lceil \eps m / 6k \rceil}$ and $y_I^+ = y_{\lfloor m - \eps m/6k\rfloor}$. 
	We say that a $\pi$-copy $t = (t_1, \ldots, t_k)$ is \emph{top-high} if $f(t_{\pi^{-1}(k)}) > y_{I_{\pi^{-1}(k)}(t)}^+$, and \emph{bottom-low} if $f(t_{\pi^{-1}(1)}) < y_{I_{\pi^{-1}(1)}(t)}^-$. 
	A copy that is neither top-high nor bottom-low is said to be \emph{non-extremal}. In other words, a $\pi$-copy is non-extremal if its highest point is not too high with respect to the interval it lies in, and similarly, its lowest point is not too low with respect to its interval.
	Note that the number of top-high $\pi$-copies in $\mathcal{C}$ is bounded by $\eps n / 6k$ (as each interval contributes no more than $\eps m / 6k$ such copies), and similarly for the number of bottom-low $\pi$-copies. Thus, there exists a set $\mathcal{C}' \subseteq \mathcal{C}$ of $\eps n / 6k$ non-extremal $\pi$-copies.
	
	The main idea is that with sufficiently many queries, the sampler -- a sample-based algorithm to find a $\pi$-copy -- will be able to capture all entries of a non-extremal $\pi$-copy $t = (t_1, \ldots, t_k) \in \mathcal{C}'$ besides the lowest entry $t_{\pi^{-1}(1)}$ and the highest entry $t_{\pi^{-1}(k)}$, which will be replaced by a small enough entry from $I_{\pi^{-1}(1)}(t)$ and a large enough entry from $I_{\pi^{-1}(k)}(t)$, respectively. Note that this is a valid $\pi$-copy.
	
	\subparagraph{Analyzing the sampler}
	Let $p = cm / \eps n = c(\eps n)^{-1/(k-1)}$ as above. Let $E_{\text{found}}$ be the event that a subsequence $g$ of $f$, constructed by putting each entry of $f$ in it with probability $p$, contains a $\pi$-copy. Using the sampling trick from the first case, it is enough to show that $\Pr(E_{\text{found}}) \geq 2/3 + 1/50$. 
	Before we continue, we define the events $A_t$, $B_t$, $E_t$ for any $\pi$-copy $t = (t_1, \ldots, t_k) \in \mathcal{C}'$ as follows. $A_t$ is the event that all entries $\{t_{\pi^{-1}(j)}\}_{j=2}^{k-1}$ of $t$ are included in $g$,
	so $\Pr(A_t) = p^{k-2} = c^{k-2} / m$.
	$B_t$ is the event that $g$ contains $x \in I_{\pi^{-1}(1)}(t)$ and $x' \in I_{\pi^{-1}(k)}(t)$ such that $f(x) \leq f(t_{\pi^{-1}}(1))$ and $f(x') \geq f(t_{\pi^{-1}}(k))$, so $\Pr(B_t) \geq 1 - 2(1 - p)^{\eps m / 6k}$.
	Finally, $E_t = A_t \cap B_t \subseteq E_{\text{found}}$ indicates that $g$ contains a $\pi$-copy.
	Note that any event $A_t$ is independent of all other events, and $B_t$ is only dependent on events $B_{t'}$ where $I_{\pi^{-1}(1)}(t) = I_{\pi^{-1}(1)}(t')$ or $I_{\pi^{-1}(k)}(t) = I_{\pi^{-1}(k)}(t')$; there are at most $2m$ such events for each $B_t$.
	
	The analysis of permutations of length $3$ differs from that of longer ones.
	
	\subparagraph{$\pi$ of length {\bf $k>3$}.}
	The probability that none of the events $A_t$ for $t \in \mathcal{C}'$ holds is at most $(1 - p^{k-2})^{|\mathcal{C}'|} \leq \exp \left( - p^{k-2} \eps n / 6k \right) = \exp(-c^{k-2} (\eps n)^{1/(k-1)} / 6k) < 1/10$. Suppose then that $A_t$ holds for some $t \in \mathcal{C'}$. 
	The probability that $B_t$ does not occur is bounded by $2(1 - p)^{\eps m / 6k} \leq 2\exp \left( -p\eps m / 6k \right) = 2\exp \left( - c m^2 / 6k n\right) = 2\exp(-\eps^{2 - 2/(k-1)} n^{1 - 2/(k-1)} c / 6k ) < 1/10$ for large enough $n$. Thus, $\Pr(E_{\text{found}}) = \Pr(\exists t \colon A_t \wedge B_t) > 8/10 > 2/3 + 1/50$ in this case, as desired.
	
	\subparagraph{$\pi$ of length $k = 3$.}
	Let $X_t$ denote the indicator random variable of $E_t$ and let $X = \sum_{t \in \mathcal{C}'} X_t$.
	For every $t \in \mathcal{C}'$, we have $\Pr(A_t) = p = c / m$ and 
	$\Pr(B_t) \geq (1 - (1-p)^{\eps m /18})^2 \geq (1 - \exp(-p \eps m / 18))^2 = (1 - e^{-c \eps / 18})^2 \geq c^2 \eps^2 / 400$, where the last inequality holds when $\eps \leq \alpha k^{-2}$ for a small enough $\alpha$, as $e^{-x} \leq 1 - 9x/10$ for small enough $x$.
	Thus, $\mathbb{E}[X] = \sum_{t \in \mathcal{C}'} \Pr(A_t) \Pr(B_t) \geq \frac{\eps n}{18} \frac{c}{m} \frac{c^2 \eps^2}{400} = \frac{c^3} {7200} \eps^{5/2} n^{1/2}$.
	On the other hand, 
	\begin{align}\label{eqn:variance_X}
	\Var(X) = \sum_{t \in \mathcal{C}'} \Var(X_t) + \sum_{t \neq t' \in \mathcal{C}} \Cov(X_t, X_{t'}) \leq \mathbb{E}[X] + 2m|\mathcal{C}'|\max_{t, t'} \Pr(A_t) \Pr(A_{t'}) \Pr(B_t),
	\end{align}
	where the inequality builds on the following two facts. The first is that $\Cov(X_t, X_{t'}) \leq \mathbb{E}(X_t X_{t'}) \leq \Pr(A_t) \Pr(A_{t'}) \Pr(B_t)$, as the events $A_t, A_{t'}, B_t$ are mutually independent. The second fact is that, for any $t \in \mathcal{C}'$, $\Cov(X_t, X_{t'}) = 0$ for all but $2m$ of the tuples $t'$, as discussed above.
	 
	The second term in \eqref{eqn:variance_X} is bounded by $2 c^2 \sqrt{\eps n}$. Therefore, the standard deviation of $X$ is bounded by $ \sqrt{\mathbb{E}[X]} + \sqrt{2} c \eps^{1/4} n^{1/4} \leq \mathbb{E}[X] / 10$, where the bound on $\eps$ in the statement of the lemma is chosen so that the last inequality holds (note that $\eps^{1/4} n^{1/4} = \eps^{5/2} n^{1/2}$ for $\eps = n^{-1/9}$, and thus the smallest possible value of $\eps$ for which this proof works is $\Theta_k(n^{-1/9})$). Using Chebyshev's inequality, $\Pr(E_{\text{found}}) = \Pr(X > 0) \geq 9/10 > 2/3 + 1/50$, concluding the proof.
\end{proof}

%% file: sec-lower.tex
In this section we provide proofs for our lower bounds in the non adaptive case. These are Theorems~\ref{thm:upper_bound_is_tight},~\ref{thm:main_thm_lower2} and~\ref{thm:entangling_to_unique}.
We start with the proof of~\autoref{thm:main_thm_lower2}. Then, we use it to finish the proof of~\autoref{thm:upper_bound_is_tight}, which requires us to prove a relatively simple special case of~\autoref{thm:entangling_to_unique}. Finally, we prove~\autoref{thm:entangling_to_unique} in its full generality. We chose to present the proofs in this order for the sake of readability, as 
the proof of~\autoref{thm:entangling_to_unique} will be easier to understand after tackling the special case considered in~\autoref{thm:upper_bound_is_tight}.

\begin{proof}[Proof of~\autoref{thm:main_thm_lower2}]
Fix a permutation $\pi$ of length $k$, and let $P = (\Lambda, S)$ be a unique signed partition of $\pi$ of size $u$, where $\Lambda = (\sigma_1, \ldots, \sigma_u)$ and $S = (s_1, \ldots, s_u)$.
Let $0 < \eps < \eps_0(k)$ and let $n > n_0(k)$ be an integer,
where $\eps_0(k) \leq 1/2k$ is small enough as a function of $k$ and $n_0(k)$ is large enough as a function of $k$. We may assume, for convenience, that $m = n/k$ is an integer and that $\eps m$ is an integer bigger than $k$ (translating the result to any $n$ and $\eps$ comes at a ``price'' of a multiplicative constant that depends only on $k$).

Recall that a one sided $\eps$-test for $\pi$-freeness must always accept
$\pi$-free sequences, and reject sequences that are $\eps$-far from $\pi$-freeness with probability at least $2/3$.
Thus, any one sided test $T$ for $\pi$-freeness must always accept its input if the subsequence it queries is $\pi$-free. This follows from the fact that for any $\pi$-free sequence $g\colon [q] \to \R$, any $n > q$ and any $1 \leq t_1 < \ldots < t_q \leq n$, there exists a $\pi$-free sequence $f\colon[n] \to \R$ such that $f(t_j) = g(j)$ for any $j = 1, \ldots, q$. That is, any $\pi$-free queried subsequence might possibly be contained in a $\pi$-free sequence, and hence must be accepted by any one-sided test.
Therefore, any one-sided test for $\pi$-freeness can be seen as a non-adaptive \emph{search algorithm for $\pi$ in $f$}, whose goal is to find a $\pi$-copy in an unknown input sequence $f$ that is guaranteed to be $\eps$-far from $\pi$-freeness, with success probability at least $2/3$. 

We use Yao's principle, constructing a family $\mathcal{F}$ of sequences $f\colon[n] \to \R$ that are $\eps$-far from $\pi$-freeness, which satisfies the following property for some constant $c_k > 0$.
For any $1 \leq t_1 < \ldots < t_q \leq n$ where $q < c_k \eps^{-1/u} n^{1 - 1/u}$, it holds that 
\begin{align}
\label{eqn:lower_bound_query}
\Pr_{f \in \mathcal{F}} \left( \text{subsequence of $f$ in indices $t_1, \ldots, t_q$ contains a $\pi$-copy} \right) < 2/3.
\end{align}
Combining~\eqref{eqn:lower_bound_query} with a standard Yao-type argument completes the proof,
as it implies that any (possibly probabilistic) search algorithm for $\pi$ in $f$, where $f$ is chosen uniformly at random from $\mathcal{F}$, must make $c_k \eps^{-1/u} n^{1 - 1/u}$ queries to have success probability $2/3$.

\paragraph{Constructing $\mathcal{F}$}
In the rest of the proof, we present a family $\mathcal{F} = \mathcal{F}(P)$ for which~\eqref{eqn:lower_bound_query} holds.
Let us describe the structure of the sequences $f \in \mathcal{F}$ before diving into the technical details. A sequence $f \in \mathcal{F}$ looks like a blowup of $f_{P}$, where each blown up part is planted, starting at a random location, inside a longer interval (making it hard for a non-adaptive test to ``guess'' where each part is located inside its interval). More specifically, 
each part $\sigma_i$ of the partition $\Lambda$ corresponds to an interval $I_i$ in $f$ whose length is $|\sigma_i| m$. In this interval, there are $\eps n$ copies of $\sigma_i$ ordered in an increasing manner if $s_i$ is a $+$, and a decreasing manner if $\sigma_i$ is a $-$, where each $\sigma_i$-copy is a consecutive subsequence of $f$. The value of $f$ on these $\sigma_i$-copies (for each $i$) is ``aligned'' with other intervals, so that $f$ contains a set of $\eps n$ pairwise disjoint $\pi$-copies, without containing any other $\pi$-copy (here we use the fact that $P = (\Lambda, S)$ is unique). The rest of the elements in each interval are chosen in a manner that does not create any other $\pi$-copy.
To make $\mathcal{F}$ ``random enough,'' the points where the consecutive copies begin in each interval are chosen uniformly at random. This assures that the probability for each $k$-tuple of entries in $f$ to induce a $\pi$-copy is sufficiently small, proving~\eqref{eqn:lower_bound_query}.

We now provide the technical details. 
For every $1 \leq i \leq u$, write $\sigma_i = \pi[j_{i-1}+1,j_i]$, where $j_0 = 0$ and let $\delta_i = j_i - j_{i-1}$ denote the length of $\sigma_i$.
For any $1 \leq i \leq u$, let $I_i = \{m j_{i-1}  + 1, \ldots, m j_i \}$.
A sequence $f \colon [n] \to \R$ is in $\mathcal{F}$ if for any $1 \leq i \leq u$ there exists $0 \leq n_i \leq (1-k \eps) \delta_i m$, such that the following conditions hold.
\begin{itemize}
	\item For every $1 \leq i \leq u$ where $s_i$ is a $+$, and every $1 \leq l \leq \delta_i$ and $0 \leq r \leq \eps n - 1$, we take $f(m j_{i-1} + n_i + r \delta_i + l) = r + \pi_{j_{i-1} + l} / 2k$. We also take $f(x) = -1$ for any $m j_{i-1} + 1 \leq x \leq m j_{i-1} + n_i$,
	and $f(x) = n$ for any $m j_{i-1} + n_i + \eps n \delta_i + 1 \leq x \leq m j_i$.
	
	\item For every $1 \leq i \leq u$ where $s_i$ is a $-$, and every $1 \leq l \leq \delta_i$ and $0 \leq r \leq \eps n - 1$, we take $f(m j_{i-1} + n_i + r \delta_i + l) = (\eps n - 1 - r) + \pi_{j_{i-1} + l} / 2k$. We also take $f(x) = n$ for any $m j_{i-1} + 1 \leq x \leq m j_{i-1} + n_i$,
	and $f(x) = -1$ for any $m j_{i-1} + n_i + \eps n \delta_i + 1 \leq x \leq m j_i$.
\end{itemize}
\paragraph{Any $f \in \mathcal{F}$ is $\eps$-far from $\pi$-freeness}
Any such $f$ is $\eps$-far from $\pi$-freeness. Indeed, for any $0 \leq r \leq \eps n - 1$, the subsequence of $f$ consisting of all $k$ entries $x \in [n]$ for which $r < f(x) < r+1$  is a $\pi$-copy, so there is a set $\mathcal{D}_f$ of $\eps n$ pairwise-disjoint $\pi$-copies in $f$. 
\paragraph{Any $f \in \mathcal{F}$ does not contain non-trivial $\pi$-copies}
On the other hand, $f$ does not contain any other (i.e., non-trivial) $\pi$-copy. To show this we use the fact that $P = (\Lambda, S)$ is a unique signed partition. 

\begin{claim}
Let $f \in \mathcal{F}$.
If $f$ contains a non-trivial $\pi$-copy, then it contains a non-trivial copy without the values $-1$ and $m$.
\end{claim}
\begin{proof}[Proof sketch]
Recall that $k < \eps m$. The proof follows by applying iteratively the following fact, and its symmetric counterpart. If $t = (t_1, \ldots, t_k)$ is a $\pi$-copy in $f$, and if there exist $0 \leq r < \eps n - 1$ and $i \in [k]$ such that $r-1 \leq f(t_i) < r$, but there is no $j \in [k]$ for which $r < f(t_j) < r+1$, then $f$ also contains a $\pi$-copy created by the following ``lifting process,'' that replaces all entries with values between $r-1$ (inclusive) and $r$ (exclusive) with entries whose values are bigger than $r$ and smaller than $r+1$. 

If $r > 0$, we replace any $t_i$ satisfying $r-1 < f(t_i) < r$ with the unique entry $t'$ satisfying $f(t') = f(t_i) + 1$. If $r = 0$ we replace $t_i$ with the closest entry $t'$ among those satisfying $0 < f(t') < 1$.
\end{proof}

Suppose now to the contrary that $f$ contains a non-trivial $\pi$-copy in the entries $x_1 < \ldots < x_k \in n$, without the values $-1$ and $m$, and let $R = \setOfSuchThat{ \lfloor f(x_i) \rfloor }{ 1 \leq i \leq k } \subseteq \{0, 1, \ldots, \eps n - 1\}$, so $2 \leq |R| \leq k$.
We now arbitrarily add elements from $\{0, 1, \ldots, \eps m - 1\}$ to $R$ to obtain a set $R'$ of size \emph{exactly} $k$.
 
Let $g$ be the subsequence of $f$ over the set of entries $W(R') = \{w \in [n] : \lfloor f(w) \rfloor \in R' \}$. In particular $x_1, \ldots, x_k \in W(R')$, so $g$ contains a non-trivial $\pi$-copy.
But this is a contradiction -- the nature of our construction (and in particular, the choice of signs) implies that $g$ is order-isomorphic to the sequence $f_P$ given in~\autoref{def:USPN}, which does not contain non-trivial $\pi$-copies (as $P$ is unique). 
Thus, the only $\pi$-copies in $f$ are the trivial copies that come from $\mathcal{D}_f$.

\paragraph{Analysis: $\mathcal{F}$ satisfies desired conditions}
Finally, we show that the probability for a $k$-tuple $1 \leq t_1 < \ldots < t_k \leq n$ 
to induce a $\pi$-copy in a sequence $f \in \mathcal{F}$ chosen uniformly at random is sufficiently small. 
We may restrict ourselves to tuples containing exactly $\delta_i$ entries in $I_i$ for any $1 \leq i \leq u$, as these are the only tuples with positive probability to induce a $\pi$-copy.
Suppose that $f \in \mathcal{F}$ contains a $\pi$-copy in entries $t_1 < \ldots < t_k$. This copy must come from $\mathcal{D}_f$, and so there exists some $ 0\leq r \leq \eps n - 1$ such that $r < f(t_l) < r+1$ for any $1 \leq l \leq k$. The values of $r$ and $t_{j_1}, t_{j_2}, \ldots, t_{j_u}$ determine $n_1, \ldots, n_u$ uniquely. In other words, $f$ is the only sequence, among all $|\mathcal{F}| > (n/2k)^{u} \geq (2k)^{-k}n^u$ sequences from $\mathcal{F}$, that has a $\pi$-copy of height between $r$ and $r+1$ whose $j_i$-th entry lies in $t_{j_i}$, for any $1 \leq i \leq u$.
In total, only at most $\eps n$ such possible choices $f \in \mathcal{F}$ have a $\pi$-copy whose $j_i$-th entry lies in $t_{j_i}$, for any $1 \leq i \leq u$. The following inequality summarizes the discussion.
\begin{equation}
\Pr_{f \in \mathcal{F}} \left( \text{subsequence of $f$ in indices $t_{j_1}, \ldots, t_{j_u}$ is contained in a $\pi$-copy} \right) \leq \frac{\eps n}{|\mathcal{F}|} < \frac{(2k)^{k} \eps}{n^{u-1}} 
\end{equation}
We are now ready to finish the proof of~\eqref{eqn:lower_bound_query}. 
Pick $c_k =(3k)^{-k/u}$, and
let $t = (t_1, \ldots, t_q)$ with $1 \leq t_1 < \ldots < t_q \leq n$ be a $q$-tuple, where $q < c_k \eps^{-1/u} n^{1 - 1/u}$. $t$ contains $\binom{q}{u} \leq q^u$ $u$-tuples, so
by a union bound, the expected number of $u$-tuples contained in a $\pi$-copy (over a uniform choice $f \in \mathcal{F}$) is less than $(2k)^k \eps q^u n^{1-u} < 2/3$. Thus, the probability that the subsequence of $f$ on $t$ contains a $\pi$-copy is less than $2/3$, as desired.
\end{proof}

\begin{proof}[Proof of~\autoref{thm:upper_bound_is_tight}]
Using~\autoref{thm:main_thm_lower2}, it is enough to show that $u(\pi) = k-1$ for any permutation $\pi$ of length $k$ satisfying $|\pi^{-1}(1) - \pi^{-1}(k)| = 1$. Let $\pi = (\pi_1, \ldots, \pi_k)$ be a permutation of length $k$, and assume, without loss of generality, that $\pi_\ell = 1$ and $\pi_{\ell+1} = k$ for some $1 \leq \ell \leq k-1$.
We take the following signed partition $P = (\Lambda, S)$ of size $k-1$. $\Lambda = (\sigma_1, \ldots, \sigma_{k-1})$ where $\sigma_i$ consists of the single element $\pi_i$ for any $i < \ell$, $\sigma_\ell = (1, k)$, and $\sigma_i$ is the single element $\pi_{i+1}$ for any $i > \ell$. The sign vector $S = (s_1, \ldots, s_{k-1})$ is defined as follows.
$s_\ell$ is a $-$, and for any $i \neq \ell$, $s_i$ is a $+$ if and only if $\pi_{i} > \pi_{i+1}$.

We now show that $P$ is unique, implying that $u(\pi) \geq |P| = k-1$, as needed.
Consider the sequence $f = f_P$, as defined in~\autoref{def:USPN}. 
We partition the entries of $f_P$ into intervals $I_1, \ldots, I_{k-1}$, where $I_i$ contains all entries that participate in the $\sigma_i$-part of some $\pi$-copy in $f_P$. In other words, $I_i = \{(i-1)k + 1 , \ldots, ik\}$ for any $1 \leq i < \ell$, $I_\ell =  \{(\ell-1)k +1, \ldots, (\ell+1)k\}$ and $I_i = \{ik+1, \ldots, (i+1)k\}$ for any $\ell < i \leq k-1$.

Let $q = (q_1, \ldots, q_\ell)$ be a $\pi$-copy in $f_P$. The following claim sheds light on the structure of $q$ with respect to the intervals $I_1, \ldots, I_{k-1}$.
\begin{claim}
\label{claim:strcutre_copy_simple_case}
For any $i = 1, \ldots, k$ let $\ind(i)$ denote the index of the interval containing $q_i$, that is, $q_i \in I_{\ind(i)}$. Then
$\ind(i) \geq i$ for any $i \leq \ell$ and $\ind(i) \leq i-1$ for any $i \geq \ell+1$.
\end{claim}
\begin{proof}
	Suppose to the contrary that $\ind(i) < i$ for some $i \leq \ell$, and consider the smallest $i$ satisfying this. Then $\ind(i-1) = \ind(i) = i-1$, that is, $q_{i-1}, q_i \in I_{i-1}$.
	This is a contradiction: If $\pi_{i} > \pi_{i-1}$ then $s_{i-1}$ is a $-$, so the subsequence of $f$ restricted to $I_{i-1}$ is decreasing, contradicting the fact that $q$ is a $\pi$-copy, that must satisfy $f(q_i) > f(q_{i-1})$ since $\pi_i > \pi_{i-1}$. If $\pi_{i} < \pi_{i-1}$ then $s_{i-1}$ is a $+$ and, symmetrically, we have a contradiction.
	Thus, $\ind(i) \geq i$ for any $i \leq \ell$. The proof that $\ind(i) \leq i-1$ for any $i \geq \ell+1$ is symmetric.
\end{proof}
As a special case of~\autoref{claim:strcutre_copy_simple_case}, we conclude that $q_\ell, q_{\ell+1} \in I_\ell$ for any $\pi$-copy $q = (q_1, \ldots, q_k)$. 
This implies that $f_P(q_\ell) = r + 1/2k$ and $f_P(q_{\ell+1}) = r + 1/2$ for some integer $r$ (since the only length-$2$ subsequences of $f$ inside $I_\ell$ that are increasing are $(r+1/2k, r+1/2)$, for any integer $0 \leq r \leq k-1$). 
Hence, for any $i \neq \ell, \ell+1$, $q_i$ must be the unique entry satisfying $f(q_i) = r + \pi_i / 2k$.
We conclude that $f_P$ does not contain non-trivial $\pi$-copies, so $P$ is unique.
\end{proof}

The proof of~\autoref{thm:entangling_to_unique} is based on ideas that are similar to those of the proof of~\autoref{thm:upper_bound_is_tight}, and in particular, a generalized form of~\autoref{claim:strcutre_copy_simple_case} serves as an important tool in the proof.
\begin{proof}[Proof of~\autoref{thm:entangling_to_unique}]
Let $\pi$ be a permutation of length $k$, and let $E = \{\tau_1, \ldots, \tau_t \}$ be an entangling 
of $\pi$ whose resulting partition $\Lambda = \Lambda(E) = (\sigma_1, \ldots, \sigma_d)$ is of size $d = d(\pi)$. 
For any $1 \leq \ell \leq t$, denote by $\lambda(\ell)$ the unique index satisfying $\tau_\ell = \sigma_{\lambda(\ell)}$.
For any $1 \leq i \leq d$, write $\sigma_i = \pi[j_{i-1}+1, j_i]$ where $j_0 = 0$ and $j_d = k$.
We choose the sign vector $S = (s_1, \ldots, s_d)$ as follows. 
\begin{itemize}
	\item For any $1 \leq i \leq d$ where $\sigma_i$ contains more than one element, $s_i$ is a $+$ if $\min{\sigma_i}$ lies after $\max{\sigma_i}$ in $\pi$, and otherwise $s_i$ is a $-$.
	\item For any $i < \lambda(1)$ where $\sigma_i$ is a single element, $s_i$ is a $+$ if and only if $\pi_{j_i} > \pi_{j_i + 1}$.
	\item For any $i > \lambda(1)$ where $\sigma_i$ is a single element, $s_i$ is a $+$ if and only if $\pi_{j_i} < \pi_{j_i - 1}$.
\end{itemize}
To finish the proof, we shall show that the signed partition $P = (\Lambda, S)$ is unique, implying that $u(\pi) \geq d = d(\pi)$. For this, we need to show that $f = f_P$ does not contain non-trivial $\pi$-copies. 
As in the proof of~\autoref{thm:upper_bound_is_tight}, for any $1 \leq i \leq d$ let $I_i = \{k j_{i-1}+1, \ldots, kj_i\}$. 
Let $q = (q_1, \ldots, q_k)$ be a $\pi$-copy in $f_P$.
The following claim is the equivalent of~\autoref{claim:strcutre_copy_simple_case} for our more general case.
\begin{claim}
	\label{claim:strcutre_copy_general_case}
	For any $i = 1, \ldots, k$ let $\ind(i)$ denote the index for which $q_i \in I_{\ind(i)}$. Then
	$\ind(j_{i-1} + 1) \geq i$ for any $i \leq \lambda(1)$ and $\ind(j_i) \leq i$ for any $i \geq \lambda(1)$.
\end{claim}
\begin{proof}
We shall prove the claim for $i \leq \lambda(1)$, as the proof for $i \geq \lambda(1)$ is symmetric.
Suppose to the contrary that there exists $i \leq \lambda(1)$ for which $\ind(j_{i-1} + 1) < i$, and consider the smallest such $i$. Then $\ind(j_{i-2} + 1) \geq i-1$, and so $\ind(j) = i-1$ for any $j_{i-2} + 1 \leq j \leq j_{i-1} + 1$. 

We show that $I_{i-1}$ does not contain a copy of $\pi[j_{i-2}+1, j_{i-1}+1]$, leading to a contradiction. If $|\sigma_{i-1}| = 1$ then the choice of the sign $s_{i-1}$ is a $+$ if $\pi_{j_{i-1}+1} < \pi_{j_{i-1}}$, and a $-$ otherwise; in the first case, the entries in $I_{i-1}$ are increasing and so it cannot contain $\pi[j_{i-i}, j_{i-1}+1]$, which is a decreasing sequence, a contradiction. In the other case we also get a contradiction, symmetrically.
Thus, from here onwards we may assume that $\sigma_{i-1}$ contains more than one element.

The choice of the sign $s_{i-1}$ implies that the only $\sigma_{i-1}$-copies in the subsequence of $f_P$ on the interval $I_{i-1}$ are the trivial ones, i.e., those that contain all $|\sigma_{i-1}|$ elements between $r$ and $r+1$ for some integer $0 \leq r \leq k-1$.
Thus, we may assume that $r < f_P(q_j) < r+1$ for any $j_{i-2} + 1 \leq j \leq j_{i-1}$.

Without loss of generality, assume that $s_{i-1}$ is a $+$; this corresponds to the case where $\max \sigma_{i-1}$ lies before $\min \sigma_{i-1}$ in $\pi$.
Since $E$ is an entangling, we know that $\sigma_{i-1}$ is not shadowed with respect to $\sigma_{\lambda(1)}$. This means that $\pi_{j_{i-1} + 1} < \max \sigma_{i-1}$, and so $f_P(q_{j_{i-1} + 1}) < r+1$. But this contradicts the fact that $\ind(j_{i-1} + 1) = i-1$: All $|\sigma_{i-1}|$ entries in $I_{i-1}$ whose value is between $r$ and $r+1$ are assigned to $q_{j_{i-2} + 1}, \ldots, q_{j_{i-1}}$,
and all entries $x$ of $I_{i-1}$ that come after these entries satisfy $f_P(x) > r+1$. In particular, $q_{j_{i-1} + 1} \in I_{i-1}$ so $f_P(q_{j_{i-1} + 1}) > r+1$, a contradiction.
\end{proof}

To show that $q$ is a trivial $\pi$-copy, we prove the following claim by induction.
\begin{claim}
	\label{claim:characterize_copies_general}
	There exists an integer $r = r(q)$ where $0 \leq r \leq k-1$, satisfying the following.
	For any $1 \leq l \leq t$, and any $j_{\lambda(\ell) - 1} + 1 \leq j \leq j_{\lambda(\ell)}$, 
	it holds that $q_j \in I_{\lambda(\ell)}$, and more specifically, $f_P(q_j) = r + \pi_j / 2k$.
\end{claim}
\begin{proof}
The proof is by induction on $\ell$. By~\autoref{claim:strcutre_copy_general_case}, $q_{j_{\lambda(1)-1}+1}, \ldots, q_{j_{\lambda(1)}} \in I_{\lambda(1)}$. By our choice of the sign $s_{\lambda(1)}$, there must be some integer $0 \leq r \leq k-1$, such that for any $j_{{\lambda(1)}-1}+1 \leq j \leq j_{\lambda(1)}$, $q_j$ is the unique entry of $f_P$ satisfying $f_P(q_j) = r + \pi_j / 2k$. This settles the case $\ell=1$. 

Now let $\ell > 1$, and assume that $f_P(q_j) = r + \pi_j / 2k$ for any $j_{\lambda(\ell') - 1} + 1 \leq j \leq j_{\lambda(\ell')}$ where $1 \leq \ell' < \ell$. We need to show that $f_P(q_j) = r + \pi_j / 2k$ for any $j_{\lambda(\ell) - 1} + 1 \leq j \leq j_{\lambda(\ell)}$.  

For any $j', j'' \in [k]$ for which we already know that $f_P(q_{j'}) =  r + \pi_{j'} / 2k$, $f_P(q_{j''}) =  r + \pi_{j''} / 2k$, and $\pi_{j'} < \pi_{j''}$, it must be true that $f_P(q_{j}) = r + \pi_j / 2k$ for any $j$ satisfying $\pi_{j'} < \pi_j < \pi_{j''}$. To see this, note that the number of entries of $f_P$ with value between $f_P(q_{j'})$ and $f_P(q_{j''})$ (not including $f_P(q_{j'}), f_P(q_{j''})$ themselves) is exactly $\pi_{j''} - \pi_{j'} - 1$. Since $q$ is a $\pi$-copy, it also contains exactly $\pi_{j''} - \pi_{j'} - 1$ entries with value between $f_P(q_{j'})$ and $f_P(q_{j''})$, so these entries of $q$ must be precisely all entries of $f_P$ whose value lies in this range.

Without loss of generality, assume that $\lambda(\ell) < \lambda(1)$ (that is, $\tau_l = \sigma_{\lambda(\ell)}$ lies before $\tau_1 = \sigma_{\lambda(1)}$ in $\pi$).
Since $E$ is an entangling, we know that $\pi_{j'} < \pi_{j_{\lambda(\ell)}} < \pi_{j''}$ for some $\pi_{j'}, \pi_{j''} \in \bigcup_{\ell' < \ell} \tau_{\ell'}$. By the previous paragraph, 
$f_P(q_{j_{\lambda(\ell)}}) = r + \pi_{j_{\lambda(\ell)}} / 2k$, also implying that $\ind(j_{\lambda(\ell)}) = \lambda(\ell)$. By~\autoref{claim:strcutre_copy_general_case}, $\ind(j_{\lambda(\ell) - 1} + 1) \geq \lambda(\ell)$, so
we get that $\ind(j) = \lambda(\ell)$ for any $j_{\lambda(\ell) - 1} + 1 \leq j \leq j_{\lambda(\ell)}$.
Considering our choice of the sign $s_l$, it follows that $f_P(q_j) = r + \pi_j / 2k$ must hold for any $j_{\lambda(\ell) - 1} + 1 \leq j \leq j_{\lambda(\ell)}$. This concludes the inductive proof.
\end{proof}

With~\autoref{claim:characterize_copies_general} it is easy to finish the proof. 
Since $E$ is an entangling, there exist $1 \leq \ell, \ell' \leq d$ such that $1 \in \tau_l = \sigma_{\lambda(\ell)}$ and $k \in \tau_{\ell'} = \sigma_{\lambda(\ell')}$, implying that $f_P(q_{\pi^{-1}(1)}) = r + 1/2k$ and $f_P(q_{\pi^{-1}(k)}) = r + 1/2$ for some $0 \leq r \leq r+1$. Thus, $r < f_P(q_j) < r+1$ for any $1 \leq j \leq k$. Since there are exactly $k$ entries $x \in [k^2]$ for which $r < f_P(x) < r+1$, $q$ must be a trivial $\pi$-copy. Therefore, $P$ is unique.
\end{proof}

We finish with an (easy) proof of~\autoref{thm:lower_bound_random_permutation} that builds on~\autoref{cor:entangling_lower_bound}.

\begin{proof}[Proof of~\autoref{thm:lower_bound_random_permutation}]
Let $\pi = (\pi_1, \ldots, \pi_k)$ be a permutation of length $k$ chosen uniformly at random. 
Without loss of generality assume that $\pi_i = 1, \pi_j = k$ for some $i < j$.
The probability that $\pi_{i+1} \leq k^{3/4}$ or $\pi_{j-1} \geq k- k^{3/4}$ or $|i - j| < k^{3/4}$ is $O(k^{-1/4})$.
Conditioning on the event that none of the above happens, the probability that there exists no $i+1 < x  < j-1$ for which $\pi_x < k^{3/4}$ and $\pi_{x+1} > k - k^{3/4}$ is also bounded by $O(k^{-1/4})$ (it is actually exponentially smaller than that). If none of these events happens, then $d(\pi) \geq k-3$, 
as there exists some $i < x < y$ for which $((\pi_x, \pi_{x+1}), (1, \pi_{i+1}), (\pi_{j-1}, k))$ is an entangling. 
Indeed, $(1, \pi_{i+1})$ and $(\pi_{j-1}, k)$ cannot be shadowed with respect to $(\pi_x, \pi_{x+1})$, and the two other conditions of an entanglement hold since $\pi_{x} < \pi_{i+1}, \pi_{j-1} < \pi_{x+1}$.
Thus $d(\pi) \geq k-3$ with probability at least $1 - O(k^{-1/4})$, as desired.

As an added bonus, note that $\Pr(d(\pi) \geq k-2) \geq 19/24 - O(1/k)$: Suppose that $i > 1$, $j < n$, and $j \geq i+2$ (all of these hold with probability $1 - O(1/k)$). Consider the event where either $\max\{\pi_{i-1}, \pi_{i+1}\} \geq \pi_{j-1}$ or $\min\{\pi_{j-1}, \pi_{j+1}\} \leq \pi_{i+1}$. This event has probability $19/24$, and if it occurs, one can verify that $d(\pi) \geq k-2$.
\end{proof}

%% file: sec-hierarchy.tex
In this section, we establish an adaptivity hierarchy theorem for testing $(1,3,2)$-freeness, by proving~\autoref{theo:hierarchy:ub} and~\autoref{theo:hierarchy:lb}. The algorithm behind the former is obtained by a rather natural modification of the adaptive test of~\cite{NewmanRabinovich2017}: We replace the only adaptive component of this test -- a variant of a binary search -- by a less query-efficient, but adaptivity-limited subroutine performing the corresponding search by a recursive partitioning of the search space. The lower bound in our case is shown by a reduction to a clean problem,~\textsc{Template-Search} (a variant of a similar problem introduced by Newman et al.,~\textsc{Intersection-Search} in the context of \emph{non-adaptive} algorithms). Showing a lower bound on~\textsc{Template-Search} for adaptive tests with $r$ rounds, however, turns out to be far from straightforward. Our inductive proof relies on a connection with anti-concentration of Binomial distributions to argue that the ``uncertainty'' left to the algorithm does not decay too fast with every stage, but rather at most by a square root (or, put differently, that the algorithm cannot restrict its search space by more than a square root at each step).

\subsection{The upper bound part of the hierarchy}
We hereby prove~\autoref{theo:hierarchy:ub}, restated below:
\theohierarchub*
\begin{proof}
As aforementioned, our upper bound relies on a modification of the adaptive algorithm of Newman et al. (\cite[Theorem 5.1]{NewmanRabinovich2017}). Before explaining our modification of this algorithm, we briefly sketch how the original works, and introduce the necessary notation.

\paragraph{The $(1,3,2)$ adaptive test of Newman et al.} The adaptive test in \cite{NewmanRabinovich2017} runs (in parallel) two different tests, call them \textsc{Test1} and \textsc{Test2}, each trying to catch a violation of a specific kind; and rejects if any of them does find a violation. (Hence, the test is clearly one-sided.) We hereafter assume that $f$ is $\eps$-far from $(1,3,2)$-free. This implies that there exists a matching $T$ of $(1,3,2)$-tuples of size at least $\eps n/3$; which can be partitioned as $T= T_1\cup T_2$, where $T_1$ contains the tuples $(i,j,k)$ such that $j-i \geq k-j$ and $T_2$ those for which $j-i < k-j$.
\begin{itemize}
  \item If $T_1$ has size at least $\abs{T_1}\geq \frac{\eps n}{6}$, then \textsc{Test1} outputs $\reject$ with probability at least $2/3$, and makes $\bigO{\log^5 n/\eps^3}$ non-adaptive queries.
  \item If $T_2$ has size at least $\abs{T_2}\geq \frac{\eps n}{6}$, then  \textsc{Test2} outputs $\reject$ with probability at least $\poly(\eps,1/\log(n))$, which can be amplified by running it independently $\poly(1/\eps,\log(n))$ times. Here too, there are two cases to consider, where in what follows 
  \begin{enumerate}
      \item $I_L,I_R\subseteq[n]$ are disjoint intervals considered by (and known to) the algorithm, where all elements of the ``left interval'' $I_L$ lie before those of the ``right interval'' $I_R$. 
      \item $I'_R\subseteq I_R$ is a subset implicitly defined as $I'_R = \setOfSuchThat{ k\in I_R }{ f(k) > \alpha }$ (where $\alpha$ is a value obtained and known by the algorithm in a previous step), and unknown to the algorithm.
  \end{enumerate}
  $I_R,I_L,\alpha$ are all induced by the behavior of the algorithm \textsc{Test2} on steps 1 to 4 of the algorithm described in \cite{NewmanRabinovich2017}; these steps perform a total of $q\eqdef\bigO{\log^{20}n/\eps^2}$ queries to $f$. We condition on these steps to be successful, which happens with probability $\bigOmega{\eps/\log^2 n}$ by the analysis of~\cite{NewmanRabinovich2017}), so that
      \begin{enumerate}
          \item If $f|_{I'_R}$ is $\frac{1}{\log^5 n}$-far from monotone (i.e., does not contain any monotone subsequence of length at least $(1-\frac{1}{\log^5 n})$), then Step 5 (which also makes $q$ queries to $f$) will output \reject with probability $1-o(1)$.
          \item If $f|_{I'_R}$ is $\frac{1}{\log^5 n}$-close to monotone, then Step 6 (which makes $O(1)$ queries to $f$) will return with probability $\bigOmega{\eps^2/\log^4 n}$ a pair $(i,j)\in I_L^2$ such that 
          \begin{enumerate}
          	\item[(a)] $i < j$ and $\alpha < f(i) < f(j)$.
          	\item[(b)] There exists $k\in I'_R$ such that $(i,j,k)\in T$. 
          \end{enumerate}
      	Then, Step 7 leverages (a), (b), and the near-monotonicity of $f|_{I'_R}$, to find some $k'\in I_R$ such that $f(k')\in(f(i),f(j))$.
      \end{enumerate}
\end{itemize}
\begin{figure}[H]\centering
    \input{fig-ub-assumptions.tex}
    \caption{Our starting assumptions for Step 7.}
\end{figure}
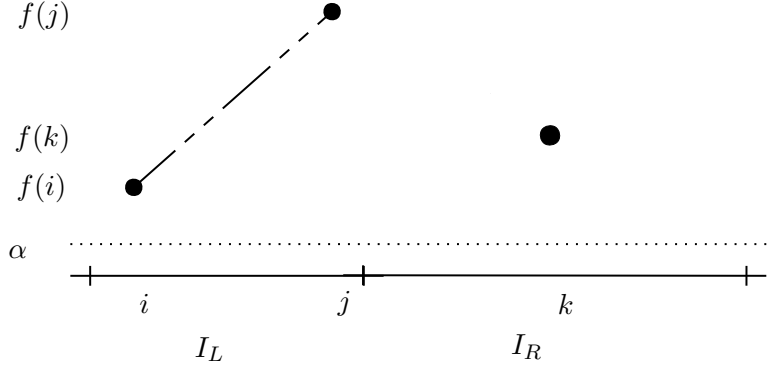
The only adaptive part of the test is Step 7 of the second case (\textsc{Test2}), and this is the one which is amenable to improvements with bounded adaptivity. Indeed, this step is implemented in~\cite{NewmanRabinovich2017} in two different ways, to obtain respectively an adaptive and a non-adaptive test: 
\begin{itemize}
	\item[(a)] An \textit{adaptive} method making a ``filtered'' binary search, with query complexity $\poly(\log n, 1/\eps)$.
	\item[(b)] A \textit{non-adaptive} sampling method, with query complexity $\bigO{(\sqrt{n}\log^2 n)/\eps}$. 
\end{itemize}
To prove our theorem, it is thus sufficient to explain how to replace the adaptive method with one that uses $r$ rounds of queries, which we do next.

\paragraph{Our goal, and high-level idea} We have to solve the following problem. We are given values $a\eqdef f(i)$, $b\eqdef f(j)$ with $a<b$, and a lower bound $\alpha \eqdef f(i_0)$ with $\alpha < a$; and granted query access to $f$ over an interval $I\eqdef I_R$ (known), with the guarantee that there exists an (unknown) subset $J\eqdef I'|_R \subseteq I$ such that:
\begin{itemize}
  \item $i\in I$ is in $J$ iff $f(i) > \alpha$; 
  \item $J$ has size at least $\delta\abs{I}$ for $\delta \eqdef \frac{\eps}{8\log n}$;
  \item $J$ contains at least $\ell$ indices $k$ such that $a < f(k) < b$ (for $\ell$ to be determined later);
  \item $f|_J$ is $\eta$-close to monotone non-decreasing, for $\eta\eqdef\frac{1}{\log^5 n}$.
\end{itemize}
The goal is to output, with non-trivial probability, either an index $k\in I$ such that $a < f(k) < b$ or two indices $k<k'\in I$ such that $\alpha < f(k') < f(k)$.

At a high-level, our algorithm acts recursively as follows. In any given round (other than the last one), our algorithm divides the interval $I_R$ into $s$ randomly chosen intervals, that partition $I'_R$ into roughly equal pieces with high probability, and queries the endpoints of these intervals. Then, using the near-monotonicity of $f|_{I'_R}$, the algorithm focuses on a single one of these intervals, which will be the $I_R$ of choice for the next round of queries. The last round just consists of sampling enough points uniformly at random, which should hit (with constant probability) at least one of the $k$ witnesses which are now concentrated in our last, much smaller interval. The issues here are that, of course: (i) we do not know $I'_R$, only $I_R$ (though testing membership in $I'_R$ is easy); (ii) that $f$ is only \emph{close} to monotone, not exactly monotone; and (iii) that $f$ is close to monotone on $I'_R$, \emph{not} on $I_R$. In spite of this, we can make the above approach work, losing only some logarithmic factors in the query complexity.

\noindent We need the following well-known fact.
\begin{proposition}[{See~\cite[Section 6.4]{DavidNagaraja:2004}}]\label{prop:uniform:sampling}
    Suppose $N-1$ points are drawn uniformly and independently at random from $[0,1]$, and let $X_1, \dots, X_N$ be the lengths of the $N$ segments they partition $[0,1]$ into. Letting $\Delta_N\eqdef \max_{1\leq k\leq N} X_l$, we have
    $
      \expect{\Delta_N} = \frac{H_N}{N}
    $ where $H_N$ is the $N$-th Harmonic number.
\end{proposition}
Combining~\autoref{prop:uniform:sampling} with Markov inequality, it follows that by taking $N= \bigO{\frac{s}{r}\log\frac{s}{r}}$ uniformly random and independent points from any ordered set $S$ (with $s=\littleO{\abs{S}}$), with probability at least $1-\frac{1}{10r}$ the partition of $S$ into $N$ ``intervals'' induced by these points is such that no interval contains more than an $\frac{1}{s}$ fraction of $S$.

The last thing we need, for our recursion to go through, is to ensure that $J$ (which has density $\delta$ in $I$) still has density roughly $\delta$ in all sub-intervals of $I$. This is not necessarily true \textit{per se}, but we can rely on~\cite[Lemma 5.3]{NewmanRabinovich2017} to get the following. There exists a subset $J'\subset J$ of size at least $\abs{J'} \geq (1-\frac{3}{\log^2 n})\abs{J}$ such that every $i\in J'$ satisfies the property below.  For every interval $K\ni i$, we have $\abs{K\cap J'} \geq \gamma \abs{K}$, for $\gamma \eqdef \eps/\log^2 n$. (In the terminology of~\cite{NewmanRabinovich2017}, no element of $J'$ is \emph{$\gamma$-deserted}.)

Then, we have the following:
\begin{itemize}
  \item $J'\subseteq J$; 
  \item $J'$ has size at least $\delta'\abs{I}$ for $\delta' \eqdef \frac{\eps}{8\log n}(1-\frac{3}{\log^2 n}) \geq \gamma$, and contains no $\gamma$-deserted elements;
  \item $J$ contains at least $\ell$ indices $k$ such that $a < f(k) < b$ (for $\ell$ to be determined later);
  \item $f|_{J'}$ is $(\eta+\gamma)$-close to monotone non-decreasing.
\end{itemize}
(The only non-immediate point is the third, which will follow from the way we guarantee the existence of these $\ell$ elements.) With all this in place, we can describe the algorithm.

\begin{algorithmic}[1]
  \State Set $s\eqdef n^{\frac{1}{r+1}}$
  \For{$1\leq t\leq r-1$, given a current interval $I_t$}
    \State Query $f$ on $N\eqdef \bigO{\frac{\log r}{\gamma}\cdot\frac{s}{r}\log\frac{s}{r}}$ uniformly random and independent points from $I_t$ 
    \State Let $S_t$ be those of these points which belong to $J$, i.e. whose $f$-values is at least $\alpha$
    \If{$S_t = \emptyset$} \Comment{Happens with negligible probability if $f$ when far from $(1,3,2)$-free}
      \State\Return $\accept$
    \EndIf
    \If{ any of these values is in $(a,b)$ } 
      \State \Return \reject \Comment{We found some $k'$ with $(f(i),f(j), f(k'))$ being an $(1,3,2)$-pattern}
    \ElsIf{the sequence corresponding to the $f$-values of $S_t$ is not non-increasing}
      \State \Return \reject. \Comment{We found some $k,k'$ with $(f(i_0),f(k'), f(k))$ being an $(1,3,2)$-pattern}
    \Else
      \State Let $I_{t+1}$ be the (only) interval induced by the points in $S_t$ such that $I_{t+1}\subseteq (f(i),f(j))$;
      \State Recurse on $I_{t+1}$.
    \EndIf
  \EndFor
\end{algorithmic}
We condition on not outputting $\accept$ during any of the $r-1$ rounds, and then on every of the rounds being such that $\abs{I_{t+1}}\leq \frac{1}{s}\abs{I_{t}}$. By a union bound, this happens with probability at least $9/10$; the first, by a union bound and the fact that $\abs{J'\cap I_t} \geq \gamma\abs{I_t}$ (from the non-desertion property) and the $O((\log r)/\gamma)$ factor in $N$. The second, again by a union bound over the $r-1$ rounds, and~\autoref{prop:uniform:sampling}.

Thus, at the end, we either have found already an $(1,3,2)$-pattern and rejected, or are left with an interval $I_{r}$ of size at most $\frac{n}{s^{r-1}} = \frac{n}{n^{\frac{r-1}{{r+1}}}} = n^{\frac{2}{r+1}}$ which contains at least $\ell$ indices $k\in I'_R\subseteq I_R$ such that $(i,j,k)\in T$ (This last point by monotonicity: all such indices must be in the remaining interval, or we would have found one already). At this stage, by taking (in the $r$th (last) round of queries) $\bigO{n^{\frac{2}{r+1}}/\ell}$ independent and uniformly distributed queries in $I_r$, we find such a $k$ with high constant probability. The total number of queries made is upper bounded by
\[
    (r-1)\cdot \left(n^{\frac{1}{r+1}}+1\right) + \bigO{n^{\frac{2}{r+1}}/\ell}
\]
which for $\ell \eqdef \tildeO{\eps n^{\frac{1}{r+1}}}$ (recalling that $r \leq \poly\log(n)$) is $\tildeO{\eps^{-1} n^{\frac{1}{r+1}}}$.

\paragraph{Last step: the promise of these $\ell$ witnesses} It only remains to describe how to achieve the guarantee of having at least $\ell$ ``witnesses'' in $J'$ (with high constant probability). This is done by calling the \textsf{DyadicSampler} (Algorithm 3.1) of~\cite{NewmanRabinovich2017} $O(\ell \eps^{-2} \log^2 n)$ times; doing so, one gets with high constant probability a set of $\ell$ many $(1,2)$ pairs, each dominating a different $(1,2)$-pair in $T_2$ (similarly as in~\cite[Section 5.2]{NewmanRabinovich2017}. By considering $(i,j)$ among these $(1,2)$ pairs such that $(f(i),f(j))$, we then have $(i,j)$ for which there exist at least $\ell$ different indices $k\in I'|_R$ such that $(i,j,k)$ is a $(1,3,2)$-pattern.

 Since $\abs{J'}=(1-o(1))\abs{I'_R}$, we can additionally guarantee at the cost of a small multiplicative factor that we get $\ell$ points in $J'$ (not only in $J=I'_R$). Further, these calls to the \textsf{DyadicSampler} can all be done in parallel, and in parallel to the first round of queries, thus preserving the number of rounds $r$.
 
\paragraph{Query and round complexity} Thus, overall the query complexity of this step is
 \[ 
    O(\ell \eps^{-2} \log^2 n)+ \tildeO{\eps^{-1} n^{\frac{1}{r+1}}} = \tildeO{\eps^{-1} n^{\frac{1}{r+1}}} + \tildeO{\eps^{-1} n^{\frac{1}{r+1}}} = \tildeO{\eps^{-1} n^{\frac{1}{r+1}}}
 \]
 as claimed, using $r$ rounds of adaptivity; and this algorithm, under the assumptions laid out when describing ``our goal,'' rejects with constant probability. Integrating it into the (non-adaptive) rest of the algorithm behind~\cite[Theorem 5.1]{NewmanRabinovich2017} yields the theorem.
\end{proof}

\subsection{The lower bound part of the hierarchy}
We now establish the lower bound part of our adaptivity hierarchy for testing $(1,3,2)$-freeness, namely~\autoref{theo:hierarchy:lb}:
\theohierarchlb*

In order to do so, we first introduce a related problem, that we refer to as \textsc{Template-Search}~--~a variant of the \textsc{Intersection-Search} problem defined in~\cite{NewmanRabinovich2017}. We then show our lower bound on the query complexity of any $r$-round algorithm solving this problem in~\autoref{sec:hierarchy:lb:templatesearch}, before showing in~\autoref{sec:hierarchy:lb:reduction} how to reduce the \textsc{Template-Search} problem to testing $(1,3,2)$-freeness, while preserving both query complexity and number of rounds of adaptivity.

\subsubsection{Lower bound on \textsc{Template-Search}}\label{sec:hierarchy:lb:templatesearch}

We start by defining the problem we consider in this subsection:
\begin{definition}[\textsc{Template-Search}]
  An instance of this problem is a tuple $(S,T)$, where $\abs{T}=m$, $\abs{S}=3m$, and $S,T$ are two non-decreasing arrays of elements from $\R$. Furthermore, $T$ is a consecutive subarray of $S$, that is there exists an integer $\Delta$ such that $S_{i+\Delta} = T_i$ for all $1\leq i\leq m$. The goal is, granted query access to both $S$ and $T$ (i.e., one can ask for the $i$-th element of either array), to determine the value $\Delta$.
\end{definition}

Our goal is to prove that this problem is ``hard'' for $r$-round algorithms, as formalized in the theorem below:
\begin{theorem}\label{theo:lb:template:search}
  Let $r\leq \log^{\bigO{1}} m$. Any (possibly randomized) $r$-round algorithm that correctly solves \textsc{Template-Search} with probability at least $2/3$ must make $\bigOmega{m^{1/2^{r+3}}}$ queries.
\end{theorem}
\begin{proof}
By (the easy direction of) Yao's principle, it is sufficient to present a particular distribution over \textsc{Template-Search} instances such that any \emph{deterministic} algorithm which succeeds with probability at least $2/3$ over a random choice of $(S,T)\sim\mathcal{D}$ must make $\bigOmega{m^{1/2^{r+3}}}$.

Our distribution $\mathcal{D}$ is as follows:
 We generate two tuples $S,T$ of elements from $\R$ and of size $\abs{S}=3m$ and $\abs{T}=m$, as follows.
\begin{itemize}
  \item $T$ is obtained by drawing $m$ numbers independently at random from $(0,1)$, and sorting them;
  \item $S$ is obtained by choosing an offset $\Delta$ uniformly at random in $\{0,\dots,2m\}$, and setting 
      \[
            S\eqdef \underbrace{(-1,\dots,-1)}_{\Delta \text{ times}} \sqcup T\sqcup \underbrace{(2,\dots,2)}_{2m-\Delta \text{ times}}
      \]
\end{itemize}
Note that with probability $1$ all elements of $T$ are distinct; we will assume this is the case in the rest of the proof.

Fix an arbitrary $r$-round algorithm $\Algo$ for \textsc{Template-Search} with query complexity $q$. For an index $i\in [m]$ (corresponding to an element $T_i$ in $T$), we write $\delta(i) = \Delta+i$ for its position in $S$, i.e. $S_{\delta(i)} = T_i$ for all $i\in[m]$ (and accordingly write $\delta(T)$ for the position of the ``template'' of $T$ inside $S$).
Moreover, for simplicity (up to losing some $\poly\log(n)$ factors in the lower bound), we can assume that the algorithm makes exactly $2q$ queries in every of the $r$ rounds, namely $q$ to $S$ and $q$ to $T$; and will proceed by induction.

We let $d_0 \eqdef m$, and $d_{\ell+1} = \alpha\frac{\sqrt{d_\ell}}{ rq^2 }$ for $0\leq \ell\leq r-1$ (where $\alpha>0$ is an absolute constant determined in the course of the analysis). First, we apply the Dvoretzky--Kiefer--Wolfowitz (DKW) inequality~\cite{DKW:56,Massart:90} to the $m$ i.i.d. samples defining $T$ to argue that, with probability at least $5/6$, the empirical distribution defined by these $m$ samples has Kolmogorov distance (i.e., maximum pointwise distance between the cumulative distribution functions) at most $O(1/\sqrt{m})$ from the uniform distribution; or, equivalently, that simultaneously every interval $[a,b]\subseteq [m]$ (recalling that $T$ is sorted) is such that $T_b - T_a$ is within an additive $O(1/\sqrt{m})$ of $\frac{1}{m}(b-a)$. We hereafter condition on this.

Consider now the $\ell$-th round, for $0\leq \ell\leq r-1$. Letting $S_\ell=\{s_i\}_{1\leq i \leq \ell q}$ and $T_\ell=\{t_i\}_{1\leq i \leq \ell q}$ be the set of queries made, in the previous rounds, to $S$ and $T$ respectively. Our induction hypothesis is that
\begin{enumerate}[$(\dagger)$]
  \item\label{assump:1:lb:adap} for all $1\leq i,j\leq \ell q$, $\abs{s_i - \delta(t_j)} \geq d_\ell$;
\end{enumerate}
In other words, this asks that no query in $S$ lies ``too close'' to the location in $S$ of a query made in $T$. (Note that as $S_0, T_0 = \emptyset$, the initial case of our induction trivially holds.) Our goal then is to show that, with high probability over the choice of $(S,T)$, we will have that~\autoref{assump:1:lb:adap} still hold at round $\ell+1$, for $d_{\ell+1} \approx \sqrt{d_\ell}$ as defined above.\medskip

Fix any pair $(s,t)\in [n]\times[m]$ that \Algo{} queries (in respectively $S$ and $T$) in the $\ell$-th round; and let $s_-,s_+\in S_\ell$, $t_-,t_+\in T_\ell$ be the indices previously queried with 
\begin{align*}
s_- &\eqdef \max\setOfSuchThat{s'\in S_\ell}{ s' < s}  & s_+ &\eqdef \max\setOfSuchThat{s'\in S_\ell}{ s' > s}\\
t_- &\eqdef \max\setOfSuchThat{t'\in T_\ell}{ t' < t}  & t_+ &\eqdef \max\setOfSuchThat{t'\in T_\ell}{ t' > t}
\end{align*}
so that $(s_-,s_+)$ and $(s_-,t_+)$ are the minimal intervals delimited by previously queried points which contain $s$ and $t$. Note that if $s_- > \Delta$ or $s_+ < \Delta$, then $s$ will fall at least as far from $\delta(T)$ as $s_-$ and $s_+$, revealing no new information; thus, we assume afterwards that $[s,s_-]\cap \delta(T) \neq \emptyset$. 
 
\noindent We divide the analysis in 3 cases, depending on the relative position of these points:
\begin{description}
  \item[ Case $s_- < \delta(t_-) < \delta(t_+) < s_+$: ] by~\autoref{assump:1:lb:adap}, we must have $\abs{\delta(t_-)-s_-}, \abs{\delta(t_+)-s_+} \geq d_\ell$.  Then, no matter where the query $t\in T$ is made, we will have $\min_{s'\in S_\ell} \abs{t-s'} = \min(\abs{t-s_-},\abs{t-s_+}) \geq d_\ell$, so our bound on $d_{\ell+1}$ is determined by $\abs{t-s}$ (i.e., by the position of the new query $s$).
  \begin{figure}[H]\centering
    \input{fig-lb-case1}
  \end{figure}
Over the choice of $S,T$, the number of elements that fall in $(s_-,\delta(t_-))$ is distributed as a Binomial random variable $X$ with parameters 
\[
N\eqdef (s_+-\delta(t_+))+(\delta(t_-)-s_-) \geq 2d_\ell, \qquad p\eqdef\frac{T_{t_-}-S_{s_-}}{(T_{t_-}-S_{s_-})+(S_{s_+}-T_{t_+})}
\] 
which has standard deviation $\sqrt{Np(1-p)}\geq \sqrt{Np/2}$.\footnote{Without loss of generality, we here assume $p\leq 1/2$; otherwise, we proceed with the same argument, but considering the number of elements in $(s_+,\delta(t_+))$ instead of $(s_-,\delta(t_-))$, and $1-p$ instead of $p$.} 

We lower bound this variance as follows.
\begin{align*}
    Np &= (s_--\delta(t_-))\cdot\left(1+\frac{\delta(t_+)-s_+}{s_--\delta(t_-)}\right)\cdot \frac{1}{1+\frac{S_{s_+}-T_{t_+}}{T_{t_-}-S_{s_-}}} \\
    &\geq d_\ell \cdot\left(1+\frac{\delta(t_+)-s_+}{s_--\delta(t_-)}\right)\cdot \frac{1}{1+\frac{S_{s_+}-T_{t_+}}{T_{t_-}-S_{s_-}}} \\
    &\geq d_\ell \cdot\left(1+\frac{\delta(t_+)-s_+}{s_--\delta(t_-)}\right)\cdot \frac{1}{1+\frac{\delta(t_+)-s_+ + O(\sqrt{m})}{s_--\delta(t_-)-O(\sqrt{m}}} \tag{DKW}\\
    &\geq d_\ell \cdot\left(1+\frac{\delta(t_+)-s_+}{s_--\delta(t_-)}\right)\cdot \frac{1}{1+\frac{2(\delta(t_+)-s_+)}{(s_--\delta(t_-))/2}} \tag{$d_\ell = \bigOmega{\sqrt{m}}$}\\
    &\geq \frac{d_\ell}{4} \tag{$\ddagger$}
\end{align*}
which holds as long as $d_\ell \geq C\sqrt{m}$ for some absolute constant $C>0$.

Since this number of elements $X$ fully characterizes, in this first case, the distance from $s$ (whose position in $(s_-,s_+)$ is known to the algorithm) to $\delta(t_-)$, $\delta(t)$, and $\delta(t_+)$ (whose relative position is the same as that of $t_-,t,t_+$ and thus also known to the algorithm), the distance of $s$ to $\delta(T)$ is equivalent to the realization of $X$. Thus, it is sufficient to show that with high probability any algorithm would have error at least $d_{\ell+1}\approx \sqrt{d_\ell}$ when guessing $X$. This follows from anticoncentration of Binomial distributions: namely, 
\[
    \inf_{x\in\R} \probaDistrOf{X}{ \abs{X-x} \leq d_{\ell+1} } = \probaDistrOf{X}{ \abs{X-\expect{X}} \leq d_{\ell+1} } 
    =  \bigO{\frac{d_{\ell+1}}{\sqrt{pN}}} \leq \frac{1}{30rq^2}
\]
the last equality by our choice of $d_{\ell+1} = \bigO{\frac{\sqrt{d_\ell}}{ rq^2 }}$, and the fact that $Np = \bigOmega{d_\ell}$.

By the above discussion, and a union bound over the 3 events, we get that 
\[
    \probaOf{ \min_{u\in\{t',t,t''\}} \abs{s-\delta(u)} \leq d_{\ell+1} } \leq \frac{1}{10rq^2}.
\]
    
  \item[ Case $\delta(t_-) < s_- < s_+ < \delta(t_+)$: ] by~\autoref{assump:1:lb:adap}, we must also have $\abs{\delta(t_-)-s_-}, \abs{\delta(t_+)-s_+} \geq d_\ell$.  
  \begin{figure}[H]\centering
    \input{fig-lb-case2}
  \end{figure}
  The situation is the same as in the first case, with the roles of $(t_-,t,t_+)$ and $(s_-,s,s_+)$ swapped. Namely, the distance of $\delta(t)$ to $s_-,s_+$ is now fully captured by the number of elements in $(\delta(t_-),s_-)$, which is (over the choice of $S,T$) distributed as a random variable $X$ following a Binomial distribution with parameters 
  \[
    N\eqdef (\delta(t_+)-s_+)+(s_--\delta(t_-)) \geq 2d_\ell, \qquad p\eqdef\frac{S_{s_-}-T_{t_-}}{(S_{s_-}-T_{t_-})+(S_{s_+}-T_{t_+})}.
  \]
  The same argument then shows that 
  \[
      \probaOf{ \min_{u\in\{s',s,s''\}} \abs{\delta(t)-u} \leq d_{\ell+1} } \leq \frac{1}{10rq^2}.
  \]
  
  \item[ Case $\delta(t_-) < s_- < \delta(t_+) < s_+$ (equivalent by symmetry to $s_- < \delta(t_-) < s_+ < \delta(t_+)$): ] by~\autoref{assump:1:lb:adap}, we must have $\abs{\delta(t_-)-s_-}, \abs{\delta(t_+)-s_-}, \abs{\delta(t_+)-s_+} \geq d_\ell$.
  \begin{figure}[H]\centering
    \input{fig-lb-case3}
  \end{figure}
  In this case, the distance of $\delta(t)$ to $s_-,s_+$ and $s$ to $\delta(s_-),\delta(s_+)$ are entirely characterized (from the point of view of \Algo) by the unknown number of elements in $(s_-,\delta(t_+))$, which is this time (over the choice of $S,T$) distributed as a random variable $X$ following a Binomial distribution with parameters 
  \[
    N\eqdef s_+-s_- = (s_+-\delta(t_+))+(\delta(t_+)-s_-) \geq 2d_\ell, \quad p\eqdef\frac{T_{t_+}-S_{s_-}}{S_{s_+}-S_{s_-}} 
    = \frac{T_{t_+}-S_{s_-}}{(S_{s_+}-T_{t_+})+(T_{t_+}-S_{s_-})}.
  \]
  The same argument as before then shows that 
  \[
      \probaOf{ \min_{\substack{u\in\{s',s,s''\}\\ v\in\{t',t,t''\}}} \abs{\delta(v)-u} \leq d_{\ell+1} } \leq \frac{1}{6rq^2}
  \]
  (where we did a union bound over the 5 events ``$s$ close to $\delta(t')$,'' ``$s$ close to $\delta(t)$,'' ``$s$ close to $\delta(t'')$,'' ``$\delta(t)$ close to $s'$,'' and ``$\delta(t)$ close to $s''$'').
\end{description}
  A union bound over all $q^2$ pairs of queries $(s,t)$ then guarantees that, with probability at least $\frac{5}{6r}$,~\autoref{assump:1:lb:adap} still holds at round $\ell+1$.
  
  This concludes the induction part of the argument; to finish the proof, observe that by a union bound over all $r$ rounds and the application of the DKW inequality, we get that~\autoref{assump:1:lb:adap} holds all through the execution with probability at least $1-(r\cdot 1/(6r)+1/6)=2/3$ (over the choice of $(S,T)$). But since the algorithm is only successful when it finds the value of $\Delta$ (i.e., when $S_r\cap \delta(T_r) \neq \emptyset$), we must have $d_{r} \leq 1$. In particular, there exists some stage $1\leq L \leq r$ such that $d_L \leq C\sqrt{m} < d_{L-1}$,\footnote{Recall that $(\ddagger)$ is only valid for $d_\ell \geq C\sqrt{m}$.}{}  which leads to 
  \[
      C\sqrt{m}\geq d_L = \frac{\alpha}{rq^2}\sqrt{d_{L-1}} = \dots 
      = \left(\frac{\alpha}{rq^2}\right)^{\sum_{a=0}^L 1/2^a}\left(d_{0}\right)^{1/2^{L+1}}
      = \left(\frac{\alpha}{rq^2}\right)^{2(1-1/2^{L+1})}m^{1/2^{L+1}}
  \]
  from which
  \[
      q^2 \geq \frac{\alpha}{r}m^{\frac{1}{2^{L+2}}} \geq \frac{\alpha}{r}m^{\frac{1}{2^{r+2}}} 
  \]
  which yields the lower bound $q = \tildeOmega{m^{1/2^{r+3}}}$.
\end{proof}

\subsubsection{Reduction from \textsc{Template-Search} to $(1,3,2)$-testing}\label{sec:hierarchy:lb:reduction}

It remains to describe and analyze a reduction from \textsc{Template-Search} to (two-sided) testing of $(1,3,2)$-freeness. We first describe the reduction, before analyzing it and establishing the required properties (i.e., that it preserves both the query complexity and the number of rounds of adaptivity).

As a first step, we note that we can without loss of generality assume any test for $(1,3,2)$-freeness to be \emph{order-based}, i.e. bases its decisions only on the relative order of the values of $f$ on its queries (and not on the values themselves). This is possible by invoking a result of Fischer~\cite{Fischer2004}, along with the fact that $\pi$-freeness is a strongly order-based property (in the terminology of~\cite{Fischer2004}).\footnote{We actually need to be relatively careful in applying the result of~\cite{Fischer2004}, as (i) we also require that the equivalence between tests and order-based tests preserve the number of rounds of adaptivity, and (ii)~\cite{Fischer2004} deals with integer-valued functions, while in our case they take values in $\R$. However, it is not hard to see by inspection of the proof of~\cite[Theorem 3.2]{Fischer2004} that the argument does preserve the number of rounds of adaptivity; as for (ii), Section 5 of Fischer's paper mentions the extension of his results to real-valued functions.}

\paragraph{Construction and simulation} Given an instance $(S,T)$ of the \textsc{Template-Search} problem with $\abs{T}=m$, we build two (random) instances $f_\yes, f_\no$ of $(1,3,2)$-freeness testing of size $n\eqdef 5m$. To obtain them, we start by describing the deterministic construction of a related function $f\colon[n]\to\R$ that both $f_\yes, f_\no$ will be based on.
\begin{enumerate}
  \item set $\delta_i \eqdef \frac{1}{4} \min(T_{i+1}-T_i, T_{i}-T_{i-1})$, for all $1\leq i\leq m-1$;
  \item define $f$ on the first $2m$ elements by $f(2i-1) = f(2i) \eqdef T_i$ for all $i\in[m]$;
  \item define $f$ on the remaining $3m$ elements by $f(i+2m) \eqdef S_i$ for all $i\in[3m]$.
\end{enumerate}
In other terms, $f$ corresponds to duplicating each element of $T$ into two adjacent identical elements, and concatenating the $2m$ resulting sequence $T'$ with $S$.

\noindent Now, we define $f_\yes$ and $f_\no$ based on this $f$, which correspond to specific element-wise perturbations of the above $f$:
\begin{enumerate}
  \item set $f_\yes(2i-1)=f_\no(2i-1) \eqdef f(2i-1)-\delta_i$ and $f_\yes(2i) = f_\no(2i) \eqdef f(2i)+\delta_i$ for all $i\in[m]$ (that is, we decrease the first copy of $T_i$ by $\delta_i$ and increase the second by $\delta_i$);
  \item set $f_\yes(i+2m) \eqdef f(i+2m)+2\delta_i$ and $f_\no(i+2m) \eqdef f(i+2m)$ for all $i\in[m]$.
\end{enumerate}
It is easy to see that $f_\yes$ is always $(1,3,2)$-free. On the other hand, our function $f_\no$ will be $\Omega(1)$-far from $(1,3,2)$-free: indeed, one can check that every $f_\no$ has exactly $m$ $(1,3,2)$-tuples (one for each element of $T$, including the two first adjacent (perturbed) copies and the counterpart in the last $3m$ elements). Thus, since ``fixing'' one such triple requires modifying one of its elements, we have that $f_\no$ is $\frac{1}{5}$-far from $(1,3,2)$-free.

Furthermore, it is straightforward to simulate query access to either $f_\yes$ or $f_\no$, provided query access to $(S,T)$, while only blowing up the number of queries by a factor $4$ (and preserving the number of rounds of adaptivity). Indeed, when an algorithm queries $f_\yes(i)$ ($f_\no$ is similar), it suffices to
\begin{itemize}
  \item query the corresponding three adjacent elements of $T$ to compute the relevant $\delta_j$;
  \item if $i>2m$, also query $S_{i-2m}$
\end{itemize}
which can be done in parallel for all queries in a given round.

\paragraph{Reduction and simulation} Assume now we have a $q$-query two-sided (order-based) test $\Tester$ for $(1,3,2)$-freeness with $r$ rounds of adaptivity, which succeeds with probability $5/6$. Given an instance $(S,T)$ of \textsc{Template-Search}, we draw and simulate access to two function $f_\yes$, $f_\no$ as above, and run $\Tester$ in parallel on both (which costs in total, per the above, at most $2\cdot 4q$ queries to $(S,T)$) \emph{on the same randomness $\omega$}. By a union bound, both instances are correct with probability at least $2/3$ over the choice of $\omega$, i.e. $\Tester$ rejects $f_\no$ while accepting $f_\yes$. Fix any $\omega$ such that this holds.

We now use the fact that $\Tester$ is order-based. Since for every $i<j$ such that $(i,j)$ is not of the form $(2\ell, \ell+\Delta+2m)$ with $\ell\in[m]$ (that is, where $i$ is the index of the second perturbed copy of an element $T_\ell$, and $j$ the index in $f$ of the element of $S$ corresponding to $T_\ell$), the order relation is the same under $f_\yes$ and $f_\no$, all the answers given to $\Tester$ for such queries will be the same under both functions. (This is by construction of $f_\yes$ and $f_\no$, and the choice of the $\delta_i$'s.)

In addition, as the two instances of the test are run on the same coin tosses, in order for them to give different answers (respectively $\accept$ and $\reject$) they must have received a different answer to the same query at some point. That is, there exist a query pair $(i,j)\in[n]^2$ queried by both instances, such that $f_\yes(i) < f_\yes(j)$ yet $f_\no(i) > f_\no(j)$. But, by the above discussion, this only happens for $(i,j)$ of the form $(2\ell, \ell+\Delta+2m)$ (or $(\ell+\Delta+2m)$ if $i>j$), from which the offset $\Delta$ can be immediately computed. 

Thus, the above simulation allows one to solve the \textsc{Template-Search} problem on any instance $(S,T)$ with probability $2/3$, while preserving the number of rounds of adaptivity $r$, and with at most $6q$ queries to $(S,T)$. By~\autoref{theo:lb:template:search}, this implies $6q = \bigOmega{m^{1/2^{r+3}}}$, establishing~\autoref{theo:hierarchy:lb}.

%% file: fig-ub-assumptions.tex
\ifx\du\undefined
  \newlength{\du}
\fi
\setlength{\du}{15\unitlength}
\begin{tikzpicture}
\pgftransformxscale{1.000000}
\pgftransformyscale{-1.000000}
\definecolor{dialinecolor}{rgb}{0.000000, 0.000000, 0.000000}
\pgfsetstrokecolor{dialinecolor}
\definecolor{dialinecolor}{rgb}{1.000000, 1.000000, 1.000000}
\pgfsetfillcolor{dialinecolor}
\pgfsetlinewidth{0.050000\du}
\pgfsetdash{}{0pt}
\pgfsetdash{}{0pt}
\pgfsetbuttcap
{
\definecolor{dialinecolor}{rgb}{0.000000, 0.000000, 0.000000}
\pgfsetfillcolor{dialinecolor}
}
\definecolor{dialinecolor}{rgb}{0.000000, 0.000000, 0.000000}
\pgfsetstrokecolor{dialinecolor}
\draw (14.150000\du,8.000000\du)--(22.050000\du,8.000000\du);
\pgfsetlinewidth{0.050000\du}
\pgfsetdash{}{0pt}
\pgfsetmiterjoin
\pgfsetbuttcap
\definecolor{dialinecolor}{rgb}{0.000000, 0.000000, 0.000000}
\pgfsetstrokecolor{dialinecolor}
\draw (14.650000\du,7.750000\du)--(14.650000\du,8.250000\du);
\pgfsetlinewidth{0.050000\du}
\pgfsetdash{}{0pt}
\pgfsetmiterjoin
\pgfsetbuttcap
\definecolor{dialinecolor}{rgb}{0.000000, 0.000000, 0.000000}
\pgfsetstrokecolor{dialinecolor}
\draw (21.550000\du,8.250000\du)--(21.550000\du,7.750000\du);
\pgfsetlinewidth{0.050000\du}
\pgfsetdash{}{0pt}
\pgfsetdash{}{0pt}
\pgfsetbuttcap
{
\definecolor{dialinecolor}{rgb}{0.000000, 0.000000, 0.000000}
\pgfsetfillcolor{dialinecolor}
}
\definecolor{dialinecolor}{rgb}{0.000000, 0.000000, 0.000000}
\pgfsetstrokecolor{dialinecolor}
\draw (21.030902\du,7.985451\du)--(31.700000\du,8.000000\du);
\pgfsetlinewidth{0.050000\du}
\pgfsetdash{}{0pt}
\pgfsetmiterjoin
\pgfsetbuttcap
\definecolor{dialinecolor}{rgb}{0.000000, 0.000000, 0.000000}
\pgfsetstrokecolor{dialinecolor}
\draw (21.531242\du,7.736133\du)--(21.530560\du,8.236132\du);
\pgfsetlinewidth{0.050000\du}
\pgfsetdash{}{0pt}
\pgfsetmiterjoin
\pgfsetbuttcap
\definecolor{dialinecolor}{rgb}{0.000000, 0.000000, 0.000000}
\pgfsetstrokecolor{dialinecolor}
\draw (31.199660\du,8.249318\du)--(31.200341\du,7.749318\du);
\pgfsetlinewidth{0.050000\du}
\pgfsetdash{{\pgflinewidth}{0.200000\du}}{0cm}
\pgfsetdash{{\pgflinewidth}{0.200000\du}}{0cm}
\pgfsetbuttcap
{
\definecolor{dialinecolor}{rgb}{0.000000, 0.000000, 0.000000}
\pgfsetfillcolor{dialinecolor}
\definecolor{dialinecolor}{rgb}{0.000000, 0.000000, 0.000000}
\pgfsetstrokecolor{dialinecolor}
\draw (14.150000\du,7.200000\du)--(31.800000\du,7.200000\du);
}
\pgfsetlinewidth{0.050000\du}
\pgfsetdash{{1.000000\du}{0.200000\du}{0.200000\du}{0.200000\du}{0.200000\du}{0.200000\du}}{0cm}
\pgfsetdash{{1.600000\du}{0.320000\du}{0.320000\du}{0.320000\du}{0.320000\du}{0.320000\du}}{0cm}
\pgfsetbuttcap
{
\definecolor{dialinecolor}{rgb}{0.000000, 0.000000, 0.000000}
\pgfsetfillcolor{dialinecolor}
}
\definecolor{dialinecolor}{rgb}{0.000000, 0.000000, 0.000000}
\pgfsetstrokecolor{dialinecolor}
\draw (15.600000\du,5.900000\du)--(20.900000\du,1.200000\du);
\pgfsetlinewidth{0.050000\du}
\pgfsetdash{}{0pt}
\pgfsetmiterjoin
\pgfsetbuttcap
\definecolor{dialinecolor}{rgb}{0.000000, 0.000000, 0.000000}
\pgfsetfillcolor{dialinecolor}
\pgfpathmoveto{\pgfpoint{15.600000\du}{5.900000\du}}
\pgfpathcurveto{\pgfpoint{15.533651\du}{5.825181\du}}{\pgfpoint{15.542121\du}{5.684014\du}}{\pgfpoint{15.616940\du}{5.617665\du}}
\pgfpathcurveto{\pgfpoint{15.691759\du}{5.551316\du}}{\pgfpoint{15.832926\du}{5.559786\du}}{\pgfpoint{15.899275\du}{5.634605\du}}
\pgfpathcurveto{\pgfpoint{15.965624\du}{5.709424\du}}{\pgfpoint{15.957154\du}{5.850591\du}}{\pgfpoint{15.882335\du}{5.916940\du}}
\pgfpathcurveto{\pgfpoint{15.807516\du}{5.983289\du}}{\pgfpoint{15.666349\du}{5.974819\du}}{\pgfpoint{15.600000\du}{5.900000\du}}
\pgfusepath{fill}
\definecolor{dialinecolor}{rgb}{0.000000, 0.000000, 0.000000}
\pgfsetstrokecolor{dialinecolor}
\pgfpathmoveto{\pgfpoint{15.600000\du}{5.900000\du}}
\pgfpathcurveto{\pgfpoint{15.533651\du}{5.825181\du}}{\pgfpoint{15.542121\du}{5.684014\du}}{\pgfpoint{15.616940\du}{5.617665\du}}
\pgfpathcurveto{\pgfpoint{15.691759\du}{5.551316\du}}{\pgfpoint{15.832926\du}{5.559786\du}}{\pgfpoint{15.899275\du}{5.634605\du}}
\pgfpathcurveto{\pgfpoint{15.965624\du}{5.709424\du}}{\pgfpoint{15.957154\du}{5.850591\du}}{\pgfpoint{15.882335\du}{5.916940\du}}
\pgfpathcurveto{\pgfpoint{15.807516\du}{5.983289\du}}{\pgfpoint{15.666349\du}{5.974819\du}}{\pgfpoint{15.600000\du}{5.900000\du}}
\pgfusepath{stroke}
\pgfsetlinewidth{0.050000\du}
\pgfsetdash{}{0pt}
\pgfsetmiterjoin
\pgfsetbuttcap
\definecolor{dialinecolor}{rgb}{0.000000, 0.000000, 0.000000}
\pgfsetfillcolor{dialinecolor}
\pgfpathmoveto{\pgfpoint{20.900000\du}{1.200000\du}}
\pgfpathcurveto{\pgfpoint{20.966349\du}{1.274819\du}}{\pgfpoint{20.957879\du}{1.415986\du}}{\pgfpoint{20.883060\du}{1.482335\du}}
\pgfpathcurveto{\pgfpoint{20.808241\du}{1.548684\du}}{\pgfpoint{20.667074\du}{1.540214\du}}{\pgfpoint{20.600725\du}{1.465395\du}}
\pgfpathcurveto{\pgfpoint{20.534376\du}{1.390576\du}}{\pgfpoint{20.542846\du}{1.249409\du}}{\pgfpoint{20.617665\du}{1.183060\du}}
\pgfpathcurveto{\pgfpoint{20.692484\du}{1.116711\du}}{\pgfpoint{20.833651\du}{1.125181\du}}{\pgfpoint{20.900000\du}{1.200000\du}}
\pgfusepath{fill}
\definecolor{dialinecolor}{rgb}{0.000000, 0.000000, 0.000000}
\pgfsetstrokecolor{dialinecolor}
\pgfpathmoveto{\pgfpoint{20.900000\du}{1.200000\du}}
\pgfpathcurveto{\pgfpoint{20.966349\du}{1.274819\du}}{\pgfpoint{20.957879\du}{1.415986\du}}{\pgfpoint{20.883060\du}{1.482335\du}}
\pgfpathcurveto{\pgfpoint{20.808241\du}{1.548684\du}}{\pgfpoint{20.667074\du}{1.540214\du}}{\pgfpoint{20.600725\du}{1.465395\du}}
\pgfpathcurveto{\pgfpoint{20.534376\du}{1.390576\du}}{\pgfpoint{20.542846\du}{1.249409\du}}{\pgfpoint{20.617665\du}{1.183060\du}}
\pgfpathcurveto{\pgfpoint{20.692484\du}{1.116711\du}}{\pgfpoint{20.833651\du}{1.125181\du}}{\pgfpoint{20.900000\du}{1.200000\du}}
\pgfusepath{stroke}
\pgfsetlinewidth{0.000000\du}
\pgfsetdash{{\pgflinewidth}{0.320000\du}}{0cm}
\pgfsetdash{{\pgflinewidth}{1.560000\du}}{0cm}
\pgfsetbuttcap
{
\definecolor{dialinecolor}{rgb}{0.000000, 0.000000, 0.000000}
\pgfsetfillcolor{dialinecolor}
\definecolor{dialinecolor}{rgb}{0.000000, 0.000000, 0.000000}
\pgfsetstrokecolor{dialinecolor}
\draw (24.750000\du,3.400000\du)--(26.450000\du,4.600000\du);
}
\definecolor{dialinecolor}{rgb}{0.000000, 0.000000, 0.000000}
\pgfsetstrokecolor{dialinecolor}
\draw (24.750000\du,3.400000\du)--(26.450000\du,4.600000\du);
\pgfsetlinewidth{0.000000\du}
\pgfsetdash{}{0pt}
\pgfsetmiterjoin
\pgfsetbuttcap
\definecolor{dialinecolor}{rgb}{0.000000, 0.000000, 0.000000}
\pgfsetfillcolor{dialinecolor}
\pgfpathmoveto{\pgfpoint{26.450000\du}{4.600000\du}}
\pgfpathcurveto{\pgfpoint{26.377915\du}{4.702121\du}}{\pgfpoint{26.203708\du}{4.732157\du}}{\pgfpoint{26.101587\du}{4.660071\du}}
\pgfpathcurveto{\pgfpoint{25.999466\du}{4.587986\du}}{\pgfpoint{25.969431\du}{4.413779\du}}{\pgfpoint{26.041516\du}{4.311658\du}}
\pgfpathcurveto{\pgfpoint{26.113601\du}{4.209537\du}}{\pgfpoint{26.287808\du}{4.179502\du}}{\pgfpoint{26.389929\du}{4.251587\du}}
\pgfpathcurveto{\pgfpoint{26.492050\du}{4.323673\du}}{\pgfpoint{26.522085\du}{4.497879\du}}{\pgfpoint{26.450000\du}{4.600000\du}}
\pgfusepath{fill}
\definecolor{dialinecolor}{rgb}{0.000000, 0.000000, 0.000000}
\pgfsetstrokecolor{dialinecolor}
\pgfpathmoveto{\pgfpoint{26.450000\du}{4.600000\du}}
\pgfpathcurveto{\pgfpoint{26.377915\du}{4.702121\du}}{\pgfpoint{26.203708\du}{4.732157\du}}{\pgfpoint{26.101587\du}{4.660071\du}}
\pgfpathcurveto{\pgfpoint{25.999466\du}{4.587986\du}}{\pgfpoint{25.969431\du}{4.413779\du}}{\pgfpoint{26.041516\du}{4.311658\du}}
\pgfpathcurveto{\pgfpoint{26.113601\du}{4.209537\du}}{\pgfpoint{26.287808\du}{4.179502\du}}{\pgfpoint{26.389929\du}{4.251587\du}}
\pgfpathcurveto{\pgfpoint{26.492050\du}{4.323673\du}}{\pgfpoint{26.522085\du}{4.497879\du}}{\pgfpoint{26.450000\du}{4.600000\du}}
\pgfusepath{stroke}
\definecolor{dialinecolor}{rgb}{0.000000, 0.000000, 0.000000}
\pgfsetstrokecolor{dialinecolor}
\node[anchor=west] at (17.050000\du,9.850000\du){$I_L$};
\definecolor{dialinecolor}{rgb}{0.000000, 0.000000, 0.000000}
\pgfsetstrokecolor{dialinecolor}
\node[anchor=west] at (25.025000\du,9.790000\du){$I_R$};
\definecolor{dialinecolor}{rgb}{0.000000, 0.000000, 0.000000}
\pgfsetstrokecolor{dialinecolor}
\node[anchor=west] at (12.550000\du,5.78000\du){$f(i)$};
\definecolor{dialinecolor}{rgb}{0.000000, 0.000000, 0.000000}
\pgfsetstrokecolor{dialinecolor}
\node[anchor=west] at (12.575000\du,1.440000\du){$f(j)$};
\definecolor{dialinecolor}{rgb}{0.000000, 0.000000, 0.000000}
\pgfsetstrokecolor{dialinecolor}
\node[anchor=west] at (12.500000\du,4.535000\du){$f(k)$};
\definecolor{dialinecolor}{rgb}{0.000000, 0.000000, 0.000000}
\pgfsetstrokecolor{dialinecolor}
\node[anchor=west] at (12.350000\du,7.400000\du){$\alpha$};
\definecolor{dialinecolor}{rgb}{0.000000, 0.000000, 0.000000}
\pgfsetstrokecolor{dialinecolor}
\node[anchor=west] at (15.650000\du,8.700000\du){$i$};
\definecolor{dialinecolor}{rgb}{0.000000, 0.000000, 0.000000}
\pgfsetstrokecolor{dialinecolor}
\node[anchor=west] at (20.650000\du,8.700000\du){$j$};
\definecolor{dialinecolor}{rgb}{0.000000, 0.000000, 0.000000}
\pgfsetstrokecolor{dialinecolor}
\node[anchor=west] at (26.200000\du,8.700000\du){$k$};
\end{tikzpicture}

%% file: fig-lb-case1.tex
\ifx\du\undefined
  \newlength{\du}
\fi
\setlength{\du}{15\unitlength}
\begin{tikzpicture}
\pgftransformxscale{1.000000}
\pgftransformyscale{-1.000000}
\definecolor{dialinecolor}{rgb}{0.000000, 0.000000, 0.000000}
\pgfsetstrokecolor{dialinecolor}
\definecolor{dialinecolor}{rgb}{1.000000, 1.000000, 1.000000}
\pgfsetfillcolor{dialinecolor}
\pgfsetlinewidth{0.100000\du}
\pgfsetdash{}{0pt}
\pgfsetdash{}{0pt}
\pgfsetbuttcap
{
\definecolor{dialinecolor}{rgb}{0.000000, 0.000000, 0.000000}
\pgfsetfillcolor{dialinecolor}
}
\definecolor{dialinecolor}{rgb}{0.000000, 0.000000, 0.000000}
\pgfsetstrokecolor{dialinecolor}
\draw (20.350000\du,3.850000\du)--(30.900000\du,3.850000\du);
\pgfsetlinewidth{0.100000\du}
\pgfsetdash{}{0pt}
\pgfsetmiterjoin
\pgfsetbuttcap
\definecolor{dialinecolor}{rgb}{0.000000, 0.000000, 0.000000}
\pgfsetfillcolor{dialinecolor}
\pgfpathmoveto{\pgfpoint{20.100000\du}{3.850000\du}}
\pgfpathcurveto{\pgfpoint{20.100000\du}{3.775000\du}}{\pgfpoint{20.175000\du}{3.700000\du}}{\pgfpoint{20.250000\du}{3.700000\du}}
\pgfpathcurveto{\pgfpoint{20.325000\du}{3.700000\du}}{\pgfpoint{20.400000\du}{3.775000\du}}{\pgfpoint{20.400000\du}{3.850000\du}}
\pgfpathcurveto{\pgfpoint{20.400000\du}{3.925000\du}}{\pgfpoint{20.325000\du}{4.000000\du}}{\pgfpoint{20.250000\du}{4.000000\du}}
\pgfpathcurveto{\pgfpoint{20.175000\du}{4.000000\du}}{\pgfpoint{20.100000\du}{3.925000\du}}{\pgfpoint{20.100000\du}{3.850000\du}}
\pgfusepath{fill}
\definecolor{dialinecolor}{rgb}{0.000000, 0.000000, 0.000000}
\pgfsetstrokecolor{dialinecolor}
\draw (20.225000\du,4.100000\du)--(20.225000\du,3.600000\du);
\pgfsetlinewidth{0.100000\du}
\pgfsetdash{}{0pt}
\pgfsetmiterjoin
\pgfsetbuttcap
\definecolor{dialinecolor}{rgb}{0.000000, 0.000000, 0.000000}
\pgfsetfillcolor{dialinecolor}
\pgfpathmoveto{\pgfpoint{31.150000\du}{3.850000\du}}
\pgfpathcurveto{\pgfpoint{31.150000\du}{3.925000\du}}{\pgfpoint{31.075000\du}{4.000000\du}}{\pgfpoint{31.000000\du}{4.000000\du}}
\pgfpathcurveto{\pgfpoint{30.925000\du}{4.000000\du}}{\pgfpoint{30.850000\du}{3.925000\du}}{\pgfpoint{30.850000\du}{3.850000\du}}
\pgfpathcurveto{\pgfpoint{30.850000\du}{3.775000\du}}{\pgfpoint{30.925000\du}{3.700000\du}}{\pgfpoint{31.000000\du}{3.700000\du}}
\pgfpathcurveto{\pgfpoint{31.075000\du}{3.700000\du}}{\pgfpoint{31.150000\du}{3.775000\du}}{\pgfpoint{31.150000\du}{3.850000\du}}
\pgfusepath{fill}
\definecolor{dialinecolor}{rgb}{0.000000, 0.000000, 0.000000}
\pgfsetstrokecolor{dialinecolor}
\draw (31.025000\du,3.600000\du)--(31.025000\du,4.100000\du);
\pgfsetlinewidth{0.100000\du}
\pgfsetdash{}{0pt}
\pgfsetdash{}{0pt}
\pgfsetbuttcap
{
\definecolor{dialinecolor}{rgb}{0.000000, 0.000000, 0.000000}
\pgfsetfillcolor{dialinecolor}
}
\definecolor{dialinecolor}{rgb}{0.000000, 0.000000, 0.000000}
\pgfsetstrokecolor{dialinecolor}
\draw (20.450000\du,6.450000\du)--(25.650000\du,6.450000\du);
\pgfsetlinewidth{0.100000\du}
\pgfsetdash{}{0pt}
\pgfsetmiterjoin
\pgfsetbuttcap
\definecolor{dialinecolor}{rgb}{0.000000, 0.000000, 0.000000}
\pgfsetfillcolor{dialinecolor}
\pgfpathmoveto{\pgfpoint{20.200000\du}{6.450000\du}}
\pgfpathcurveto{\pgfpoint{20.200000\du}{6.375000\du}}{\pgfpoint{20.275000\du}{6.300000\du}}{\pgfpoint{20.350000\du}{6.300000\du}}
\pgfpathcurveto{\pgfpoint{20.425000\du}{6.300000\du}}{\pgfpoint{20.500000\du}{6.375000\du}}{\pgfpoint{20.500000\du}{6.450000\du}}
\pgfpathcurveto{\pgfpoint{20.500000\du}{6.525000\du}}{\pgfpoint{20.425000\du}{6.600000\du}}{\pgfpoint{20.350000\du}{6.600000\du}}
\pgfpathcurveto{\pgfpoint{20.275000\du}{6.600000\du}}{\pgfpoint{20.200000\du}{6.525000\du}}{\pgfpoint{20.200000\du}{6.450000\du}}
\pgfusepath{fill}
\definecolor{dialinecolor}{rgb}{0.000000, 0.000000, 0.000000}
\pgfsetstrokecolor{dialinecolor}
\draw (20.325000\du,6.700000\du)--(20.325000\du,6.200000\du);
\pgfsetlinewidth{0.100000\du}
\pgfsetdash{}{0pt}
\pgfsetmiterjoin
\pgfsetbuttcap
\definecolor{dialinecolor}{rgb}{0.000000, 0.000000, 0.000000}
\pgfsetfillcolor{dialinecolor}
\pgfpathmoveto{\pgfpoint{25.900000\du}{6.450000\du}}
\pgfpathcurveto{\pgfpoint{25.900000\du}{6.525000\du}}{\pgfpoint{25.825000\du}{6.600000\du}}{\pgfpoint{25.750000\du}{6.600000\du}}
\pgfpathcurveto{\pgfpoint{25.675000\du}{6.600000\du}}{\pgfpoint{25.600000\du}{6.525000\du}}{\pgfpoint{25.600000\du}{6.450000\du}}
\pgfpathcurveto{\pgfpoint{25.600000\du}{6.375000\du}}{\pgfpoint{25.675000\du}{6.300000\du}}{\pgfpoint{25.750000\du}{6.300000\du}}
\pgfpathcurveto{\pgfpoint{25.825000\du}{6.300000\du}}{\pgfpoint{25.900000\du}{6.375000\du}}{\pgfpoint{25.900000\du}{6.450000\du}}
\pgfusepath{fill}
\definecolor{dialinecolor}{rgb}{0.000000, 0.000000, 0.000000}
\pgfsetstrokecolor{dialinecolor}
\draw (25.775000\du,6.200000\du)--(25.775000\du,6.700000\du);
\pgfsetlinewidth{0.100000\du}
\pgfsetdash{}{0pt}
\pgfsetdash{}{0pt}
\pgfsetbuttcap
{
\definecolor{dialinecolor}{rgb}{0.000000, 0.000000, 0.000000}
\pgfsetfillcolor{dialinecolor}
}
\definecolor{dialinecolor}{rgb}{0.000000, 0.000000, 0.000000}
\pgfsetstrokecolor{dialinecolor}
\draw (22.995700\du,3.965700\du)--(28.195700\du,3.965700\du);
\pgfsetlinewidth{0.100000\du}
\pgfsetdash{}{0pt}
\pgfsetmiterjoin
\pgfsetbuttcap
\definecolor{dialinecolor}{rgb}{0.000000, 0.000000, 0.000000}
\pgfsetfillcolor{dialinecolor}
\pgfpathmoveto{\pgfpoint{22.745700\du}{3.965700\du}}
\pgfpathcurveto{\pgfpoint{22.745700\du}{3.890700\du}}{\pgfpoint{22.820700\du}{3.815700\du}}{\pgfpoint{22.895700\du}{3.815700\du}}
\pgfpathcurveto{\pgfpoint{22.970700\du}{3.815700\du}}{\pgfpoint{23.045700\du}{3.890700\du}}{\pgfpoint{23.045700\du}{3.965700\du}}
\pgfpathcurveto{\pgfpoint{23.045700\du}{4.040700\du}}{\pgfpoint{22.970700\du}{4.115700\du}}{\pgfpoint{22.895700\du}{4.115700\du}}
\pgfpathcurveto{\pgfpoint{22.820700\du}{4.115700\du}}{\pgfpoint{22.745700\du}{4.040700\du}}{\pgfpoint{22.745700\du}{3.965700\du}}
\pgfusepath{fill}
\definecolor{dialinecolor}{rgb}{0.000000, 0.000000, 0.000000}
\pgfsetstrokecolor{dialinecolor}
\draw (22.870700\du,4.215700\du)--(22.870700\du,3.715700\du);
\pgfsetlinewidth{0.100000\du}
\pgfsetdash{}{0pt}
\pgfsetmiterjoin
\pgfsetbuttcap
\definecolor{dialinecolor}{rgb}{0.000000, 0.000000, 0.000000}
\pgfsetfillcolor{dialinecolor}
\pgfpathmoveto{\pgfpoint{28.445700\du}{3.965700\du}}
\pgfpathcurveto{\pgfpoint{28.445700\du}{4.040700\du}}{\pgfpoint{28.370700\du}{4.115700\du}}{\pgfpoint{28.295700\du}{4.115700\du}}
\pgfpathcurveto{\pgfpoint{28.220700\du}{4.115700\du}}{\pgfpoint{28.145700\du}{4.040700\du}}{\pgfpoint{28.145700\du}{3.965700\du}}
\pgfpathcurveto{\pgfpoint{28.145700\du}{3.890700\du}}{\pgfpoint{28.220700\du}{3.815700\du}}{\pgfpoint{28.295700\du}{3.815700\du}}
\pgfpathcurveto{\pgfpoint{28.370700\du}{3.815700\du}}{\pgfpoint{28.445700\du}{3.890700\du}}{\pgfpoint{28.445700\du}{3.965700\du}}
\pgfusepath{fill}
\definecolor{dialinecolor}{rgb}{0.000000, 0.000000, 0.000000}
\pgfsetstrokecolor{dialinecolor}
\draw (28.320700\du,3.715700\du)--(28.320700\du,4.215700\du);
\definecolor{dialinecolor}{rgb}{0.000000, 0.000000, 0.000000}
\pgfsetstrokecolor{dialinecolor}
\node[anchor=west] at (19.450000\du,2.900000\du){$s_-$};
\definecolor{dialinecolor}{rgb}{0.000000, 0.000000, 0.000000}
\pgfsetstrokecolor{dialinecolor}
\node[anchor=west] at (30.225000\du,2.790000\du){$s_+$};
\definecolor{dialinecolor}{rgb}{0.000000, 0.000000, 0.000000}
\pgfsetstrokecolor{dialinecolor}
\node[anchor=west] at (19.700000\du,7.485000\du){$t_-$};
\definecolor{dialinecolor}{rgb}{0.000000, 0.000000, 0.000000}
\pgfsetstrokecolor{dialinecolor}
\node[anchor=west] at (25.025000\du,7.530000\du){$t_+$};
\pgfsetlinewidth{0.050000\du}
\pgfsetdash{{\pgflinewidth}{0.200000\du}}{0cm}
\pgfsetdash{{\pgflinewidth}{0.200000\du}}{0cm}
\pgfsetbuttcap
{
\definecolor{dialinecolor}{rgb}{0.000000, 0.000000, 0.000000}
\pgfsetfillcolor{dialinecolor}
\definecolor{dialinecolor}{rgb}{0.000000, 0.000000, 0.000000}
\pgfsetstrokecolor{dialinecolor}
\draw (25.700000\du,6.500000\du)--(28.350000\du,3.900000\du);
}
\pgfsetlinewidth{0.050000\du}
\pgfsetdash{{\pgflinewidth}{0.200000\du}}{0cm}
\pgfsetdash{{\pgflinewidth}{0.200000\du}}{0cm}
\pgfsetbuttcap
{
\definecolor{dialinecolor}{rgb}{0.000000, 0.000000, 0.000000}
\pgfsetfillcolor{dialinecolor}
\definecolor{dialinecolor}{rgb}{0.000000, 0.000000, 0.000000}
\pgfsetstrokecolor{dialinecolor}
\draw (20.260400\du,6.480400\du)--(22.910400\du,3.880400\du);
}
\definecolor{dialinecolor}{rgb}{0.000000, 0.000000, 0.000000}
\pgfsetstrokecolor{dialinecolor}
\node[anchor=west] at (16.750000\du,3.950000\du){$S$};
\definecolor{dialinecolor}{rgb}{0.000000, 0.000000, 0.000000}
\pgfsetstrokecolor{dialinecolor}
\node[anchor=west] at (16.800000\du,6.650000\du){$T$};
\end{tikzpicture}

%% file: fig-lb-case2.tex
\ifx\du\undefined
  \newlength{\du}
\fi
\setlength{\du}{15\unitlength}
\begin{tikzpicture}
\pgftransformxscale{1.000000}
\pgftransformyscale{-1.000000}
\definecolor{dialinecolor}{rgb}{0.000000, 0.000000, 0.000000}
\pgfsetstrokecolor{dialinecolor}
\definecolor{dialinecolor}{rgb}{1.000000, 1.000000, 1.000000}
\pgfsetfillcolor{dialinecolor}
\pgfsetlinewidth{0.100000\du}
\pgfsetdash{}{0pt}
\pgfsetdash{}{0pt}
\pgfsetbuttcap
{
\definecolor{dialinecolor}{rgb}{0.000000, 0.000000, 0.000000}
\pgfsetfillcolor{dialinecolor}
}
\definecolor{dialinecolor}{rgb}{0.000000, 0.000000, 0.000000}
\pgfsetstrokecolor{dialinecolor}
\draw (23.762499\du,1.888054\du)--(28.900001\du,1.899446\du);
\pgfsetlinewidth{0.100000\du}
\pgfsetdash{}{0pt}
\pgfsetmiterjoin
\pgfsetbuttcap
\definecolor{dialinecolor}{rgb}{0.000000, 0.000000, 0.000000}
\pgfsetfillcolor{dialinecolor}
\pgfpathmoveto{\pgfpoint{23.512500\du}{1.887500\du}}
\pgfpathcurveto{\pgfpoint{23.512666\du}{1.812500\du}}{\pgfpoint{23.587832\du}{1.737667\du}}{\pgfpoint{23.662832\du}{1.737833\du}}
\pgfpathcurveto{\pgfpoint{23.737832\du}{1.737999\du}}{\pgfpoint{23.812666\du}{1.813165\du}}{\pgfpoint{23.812499\du}{1.888165\du}}
\pgfpathcurveto{\pgfpoint{23.812333\du}{1.963165\du}}{\pgfpoint{23.737167\du}{2.037999\du}}{\pgfpoint{23.662167\du}{2.037832\du}}
\pgfpathcurveto{\pgfpoint{23.587167\du}{2.037666\du}}{\pgfpoint{23.512334\du}{1.962500\du}}{\pgfpoint{23.512500\du}{1.887500\du}}
\pgfusepath{fill}
\definecolor{dialinecolor}{rgb}{0.000000, 0.000000, 0.000000}
\pgfsetstrokecolor{dialinecolor}
\draw (23.636945\du,2.137777\du)--(23.638054\du,1.637778\du);
\pgfsetlinewidth{0.100000\du}
\pgfsetdash{}{0pt}
\pgfsetmiterjoin
\pgfsetbuttcap
\definecolor{dialinecolor}{rgb}{0.000000, 0.000000, 0.000000}
\pgfsetfillcolor{dialinecolor}
\pgfpathmoveto{\pgfpoint{29.150000\du}{1.900000\du}}
\pgfpathcurveto{\pgfpoint{29.149834\du}{1.975000\du}}{\pgfpoint{29.074668\du}{2.049833\du}}{\pgfpoint{28.999668\du}{2.049667\du}}
\pgfpathcurveto{\pgfpoint{28.924668\du}{2.049501\du}}{\pgfpoint{28.849834\du}{1.974335\du}}{\pgfpoint{28.850001\du}{1.899335\du}}
\pgfpathcurveto{\pgfpoint{28.850167\du}{1.824335\du}}{\pgfpoint{28.925333\du}{1.749501\du}}{\pgfpoint{29.000333\du}{1.749668\du}}
\pgfpathcurveto{\pgfpoint{29.075333\du}{1.749834\du}}{\pgfpoint{29.150166\du}{1.825000\du}}{\pgfpoint{29.150000\du}{1.900000\du}}
\pgfusepath{fill}
\definecolor{dialinecolor}{rgb}{0.000000, 0.000000, 0.000000}
\pgfsetstrokecolor{dialinecolor}
\draw (29.025555\du,1.649723\du)--(29.024446\du,2.149722\du);
\pgfsetlinewidth{0.100000\du}
\pgfsetdash{}{0pt}
\pgfsetdash{}{0pt}
\pgfsetbuttcap
{
\definecolor{dialinecolor}{rgb}{0.000000, 0.000000, 0.000000}
\pgfsetfillcolor{dialinecolor}
}
\definecolor{dialinecolor}{rgb}{0.000000, 0.000000, 0.000000}
\pgfsetstrokecolor{dialinecolor}
\draw (21.745699\du,1.965064\du)--(32.337501\du,1.938136\du);
\pgfsetlinewidth{0.100000\du}
\pgfsetdash{}{0pt}
\pgfsetmiterjoin
\pgfsetbuttcap
\definecolor{dialinecolor}{rgb}{0.000000, 0.000000, 0.000000}
\pgfsetfillcolor{dialinecolor}
\pgfpathmoveto{\pgfpoint{21.495700\du}{1.965700\du}}
\pgfpathcurveto{\pgfpoint{21.495509\du}{1.890700\du}}{\pgfpoint{21.570318\du}{1.815510\du}}{\pgfpoint{21.645318\du}{1.815319\du}}
\pgfpathcurveto{\pgfpoint{21.720318\du}{1.815128\du}}{\pgfpoint{21.795508\du}{1.889938\du}}{\pgfpoint{21.795699\du}{1.964937\du}}
\pgfpathcurveto{\pgfpoint{21.795890\du}{2.039937\du}}{\pgfpoint{21.721081\du}{2.115127\du}}{\pgfpoint{21.646081\du}{2.115318\du}}
\pgfpathcurveto{\pgfpoint{21.571081\du}{2.115509\du}}{\pgfpoint{21.495891\du}{2.040700\du}}{\pgfpoint{21.495700\du}{1.965700\du}}
\pgfusepath{fill}
\definecolor{dialinecolor}{rgb}{0.000000, 0.000000, 0.000000}
\pgfsetstrokecolor{dialinecolor}
\draw (21.621335\du,2.215381\du)--(21.620064\du,1.715383\du);
\pgfsetlinewidth{0.100000\du}
\pgfsetdash{}{0pt}
\pgfsetmiterjoin
\pgfsetbuttcap
\definecolor{dialinecolor}{rgb}{0.000000, 0.000000, 0.000000}
\pgfsetfillcolor{dialinecolor}
\pgfpathmoveto{\pgfpoint{32.587500\du}{1.937500\du}}
\pgfpathcurveto{\pgfpoint{32.587691\du}{2.012500\du}}{\pgfpoint{32.512882\du}{2.087690\du}}{\pgfpoint{32.437882\du}{2.087881\du}}
\pgfpathcurveto{\pgfpoint{32.362882\du}{2.088072\du}}{\pgfpoint{32.287692\du}{2.013262\du}}{\pgfpoint{32.287501\du}{1.938263\du}}
\pgfpathcurveto{\pgfpoint{32.287310\du}{1.863263\du}}{\pgfpoint{32.362119\du}{1.788073\du}}{\pgfpoint{32.437119\du}{1.787882\du}}
\pgfpathcurveto{\pgfpoint{32.512119\du}{1.787691\du}}{\pgfpoint{32.587309\du}{1.862500\du}}{\pgfpoint{32.587500\du}{1.937500\du}}
\pgfusepath{fill}
\definecolor{dialinecolor}{rgb}{0.000000, 0.000000, 0.000000}
\pgfsetstrokecolor{dialinecolor}
\draw (32.461865\du,1.687819\du)--(32.463136\du,2.187817\du);
\definecolor{dialinecolor}{rgb}{0.000000, 0.000000, 0.000000}
\pgfsetstrokecolor{dialinecolor}
\node[anchor=west] at (22.950000\du,1.250000\du){$s_-$};
\definecolor{dialinecolor}{rgb}{0.000000, 0.000000, 0.000000}
\pgfsetstrokecolor{dialinecolor}
\node[anchor=west] at (28.225000\du,1.290000\du){$s_+$};
\definecolor{dialinecolor}{rgb}{0.000000, 0.000000, 0.000000}
\pgfsetstrokecolor{dialinecolor}
\node[anchor=west] at (18.450000\du,5.485000\du){$t_-$};
\definecolor{dialinecolor}{rgb}{0.000000, 0.000000, 0.000000}
\pgfsetstrokecolor{dialinecolor}
\node[anchor=west] at (28.525000\du,5.530000\du){$t_+$};
\pgfsetlinewidth{0.050000\du}
\pgfsetdash{{\pgflinewidth}{0.200000\du}}{0cm}
\pgfsetdash{{\pgflinewidth}{0.200000\du}}{0cm}
\pgfsetbuttcap
{
\definecolor{dialinecolor}{rgb}{0.000000, 0.000000, 0.000000}
\pgfsetfillcolor{dialinecolor}
\definecolor{dialinecolor}{rgb}{0.000000, 0.000000, 0.000000}
\pgfsetstrokecolor{dialinecolor}
\draw (19.010400\du,4.480400\du)--(21.660400\du,1.880400\du);
}
\definecolor{dialinecolor}{rgb}{0.000000, 0.000000, 0.000000}
\pgfsetstrokecolor{dialinecolor}
\node[anchor=west] at (16.750000\du,1.950000\du){$S$};
\definecolor{dialinecolor}{rgb}{0.000000, 0.000000, 0.000000}
\pgfsetstrokecolor{dialinecolor}
\node[anchor=west] at (16.800000\du,4.650000\du){$T$};
\pgfsetlinewidth{0.100000\du}
\pgfsetdash{}{0pt}
\pgfsetdash{}{0pt}
\pgfsetbuttcap
{
\definecolor{dialinecolor}{rgb}{0.000000, 0.000000, 0.000000}
\pgfsetfillcolor{dialinecolor}
}
\definecolor{dialinecolor}{rgb}{0.000000, 0.000000, 0.000000}
\pgfsetstrokecolor{dialinecolor}
\draw (19.095699\du,4.457664\du)--(29.687501\du,4.430736\du);
\pgfsetlinewidth{0.100000\du}
\pgfsetdash{}{0pt}
\pgfsetmiterjoin
\pgfsetbuttcap
\definecolor{dialinecolor}{rgb}{0.000000, 0.000000, 0.000000}
\pgfsetfillcolor{dialinecolor}
\pgfpathmoveto{\pgfpoint{18.845700\du}{4.458300\du}}
\pgfpathcurveto{\pgfpoint{18.845509\du}{4.383300\du}}{\pgfpoint{18.920318\du}{4.308110\du}}{\pgfpoint{18.995318\du}{4.307919\du}}
\pgfpathcurveto{\pgfpoint{19.070318\du}{4.307728\du}}{\pgfpoint{19.145508\du}{4.382538\du}}{\pgfpoint{19.145699\du}{4.457537\du}}
\pgfpathcurveto{\pgfpoint{19.145890\du}{4.532537\du}}{\pgfpoint{19.071081\du}{4.607727\du}}{\pgfpoint{18.996081\du}{4.607918\du}}
\pgfpathcurveto{\pgfpoint{18.921081\du}{4.608109\du}}{\pgfpoint{18.845891\du}{4.533300\du}}{\pgfpoint{18.845700\du}{4.458300\du}}
\pgfusepath{fill}
\definecolor{dialinecolor}{rgb}{0.000000, 0.000000, 0.000000}
\pgfsetstrokecolor{dialinecolor}
\draw (18.971335\du,4.707981\du)--(18.970064\du,4.207983\du);
\pgfsetlinewidth{0.100000\du}
\pgfsetdash{}{0pt}
\pgfsetmiterjoin
\pgfsetbuttcap
\definecolor{dialinecolor}{rgb}{0.000000, 0.000000, 0.000000}
\pgfsetfillcolor{dialinecolor}
\pgfpathmoveto{\pgfpoint{29.937500\du}{4.430100\du}}
\pgfpathcurveto{\pgfpoint{29.937691\du}{4.505100\du}}{\pgfpoint{29.862882\du}{4.580290\du}}{\pgfpoint{29.787882\du}{4.580481\du}}
\pgfpathcurveto{\pgfpoint{29.712882\du}{4.580672\du}}{\pgfpoint{29.637692\du}{4.505862\du}}{\pgfpoint{29.637501\du}{4.430863\du}}
\pgfpathcurveto{\pgfpoint{29.637310\du}{4.355863\du}}{\pgfpoint{29.712119\du}{4.280673\du}}{\pgfpoint{29.787119\du}{4.280482\du}}
\pgfpathcurveto{\pgfpoint{29.862119\du}{4.280291\du}}{\pgfpoint{29.937309\du}{4.355100\du}}{\pgfpoint{29.937500\du}{4.430100\du}}
\pgfusepath{fill}
\definecolor{dialinecolor}{rgb}{0.000000, 0.000000, 0.000000}
\pgfsetstrokecolor{dialinecolor}
\draw (29.811865\du,4.180419\du)--(29.813136\du,4.680417\du);
\pgfsetlinewidth{0.050000\du}
\pgfsetdash{{\pgflinewidth}{0.200000\du}}{0cm}
\pgfsetdash{{\pgflinewidth}{0.200000\du}}{0cm}
\pgfsetbuttcap
{
\definecolor{dialinecolor}{rgb}{0.000000, 0.000000, 0.000000}
\pgfsetfillcolor{dialinecolor}
\definecolor{dialinecolor}{rgb}{0.000000, 0.000000, 0.000000}
\pgfsetstrokecolor{dialinecolor}
\draw (29.810400\du,4.495400\du)--(32.460400\du,1.895400\du);
}
\end{tikzpicture}

%% file: fig-lb-case3.tex
\ifx\du\undefined
  \newlength{\du}
\fi
\setlength{\du}{15\unitlength}
\begin{tikzpicture}
\pgftransformxscale{1.000000}
\pgftransformyscale{-1.000000}
\definecolor{dialinecolor}{rgb}{0.000000, 0.000000, 0.000000}
\pgfsetstrokecolor{dialinecolor}
\definecolor{dialinecolor}{rgb}{1.000000, 1.000000, 1.000000}
\pgfsetfillcolor{dialinecolor}
\pgfsetlinewidth{0.100000\du}
\pgfsetdash{}{0pt}
\pgfsetdash{}{0pt}
\pgfsetbuttcap
{
\definecolor{dialinecolor}{rgb}{0.000000, 0.000000, 0.000000}
\pgfsetfillcolor{dialinecolor}
}
\definecolor{dialinecolor}{rgb}{0.000000, 0.000000, 0.000000}
\pgfsetstrokecolor{dialinecolor}
\draw (23.787499\du,2.186809\du)--(32.337501\du,2.163191\du);
\pgfsetlinewidth{0.100000\du}
\pgfsetdash{}{0pt}
\pgfsetmiterjoin
\pgfsetbuttcap
\definecolor{dialinecolor}{rgb}{0.000000, 0.000000, 0.000000}
\pgfsetfillcolor{dialinecolor}
\pgfpathmoveto{\pgfpoint{23.537500\du}{2.187500\du}}
\pgfpathcurveto{\pgfpoint{23.537293\du}{2.112500\du}}{\pgfpoint{23.612085\du}{2.037293\du}}{\pgfpoint{23.687085\du}{2.037086\du}}
\pgfpathcurveto{\pgfpoint{23.762085\du}{2.036879\du}}{\pgfpoint{23.837292\du}{2.111672\du}}{\pgfpoint{23.837499\du}{2.186671\du}}
\pgfpathcurveto{\pgfpoint{23.837706\du}{2.261671\du}}{\pgfpoint{23.762914\du}{2.336878\du}}{\pgfpoint{23.687914\du}{2.337085\du}}
\pgfpathcurveto{\pgfpoint{23.612914\du}{2.337292\du}}{\pgfpoint{23.537707\du}{2.262500\du}}{\pgfpoint{23.537500\du}{2.187500\du}}
\pgfusepath{fill}
\definecolor{dialinecolor}{rgb}{0.000000, 0.000000, 0.000000}
\pgfsetstrokecolor{dialinecolor}
\draw (23.663190\du,2.437154\du)--(23.661809\du,1.937156\du);
\pgfsetlinewidth{0.100000\du}
\pgfsetdash{}{0pt}
\pgfsetmiterjoin
\pgfsetbuttcap
\definecolor{dialinecolor}{rgb}{0.000000, 0.000000, 0.000000}
\pgfsetfillcolor{dialinecolor}
\pgfpathmoveto{\pgfpoint{32.587500\du}{2.162500\du}}
\pgfpathcurveto{\pgfpoint{32.587707\du}{2.237500\du}}{\pgfpoint{32.512915\du}{2.312707\du}}{\pgfpoint{32.437915\du}{2.312914\du}}
\pgfpathcurveto{\pgfpoint{32.362915\du}{2.313121\du}}{\pgfpoint{32.287708\du}{2.238328\du}}{\pgfpoint{32.287501\du}{2.163329\du}}
\pgfpathcurveto{\pgfpoint{32.287294\du}{2.088329\du}}{\pgfpoint{32.362086\du}{2.013122\du}}{\pgfpoint{32.437086\du}{2.012915\du}}
\pgfpathcurveto{\pgfpoint{32.512086\du}{2.012708\du}}{\pgfpoint{32.587293\du}{2.087500\du}}{\pgfpoint{32.587500\du}{2.162500\du}}
\pgfusepath{fill}
\definecolor{dialinecolor}{rgb}{0.000000, 0.000000, 0.000000}
\pgfsetstrokecolor{dialinecolor}
\draw (32.461810\du,1.912846\du)--(32.463191\du,2.412844\du);
\pgfsetlinewidth{0.100000\du}
\pgfsetdash{}{0pt}
\pgfsetdash{}{0pt}
\pgfsetbuttcap
{
\definecolor{dialinecolor}{rgb}{0.000000, 0.000000, 0.000000}
\pgfsetfillcolor{dialinecolor}
}
\definecolor{dialinecolor}{rgb}{0.000000, 0.000000, 0.000000}
\pgfsetstrokecolor{dialinecolor}
\draw (21.745699\du,2.216506\du)--(28.087501\du,2.236704\du);
\pgfsetlinewidth{0.100000\du}
\pgfsetdash{}{0pt}
\pgfsetmiterjoin
\pgfsetbuttcap
\definecolor{dialinecolor}{rgb}{0.000000, 0.000000, 0.000000}
\pgfsetfillcolor{dialinecolor}
\pgfpathmoveto{\pgfpoint{21.495700\du}{2.215710\du}}
\pgfpathcurveto{\pgfpoint{21.495939\du}{2.140710\du}}{\pgfpoint{21.571177\du}{2.065950\du}}{\pgfpoint{21.646177\du}{2.066188\du}}
\pgfpathcurveto{\pgfpoint{21.721177\du}{2.066427\du}}{\pgfpoint{21.795937\du}{2.141666\du}}{\pgfpoint{21.795698\du}{2.216665\du}}
\pgfpathcurveto{\pgfpoint{21.795460\du}{2.291665\du}}{\pgfpoint{21.720221\du}{2.366426\du}}{\pgfpoint{21.645222\du}{2.366187\du}}
\pgfpathcurveto{\pgfpoint{21.570222\du}{2.365948\du}}{\pgfpoint{21.495461\du}{2.290710\du}}{\pgfpoint{21.495700\du}{2.215710\du}}
\pgfusepath{fill}
\definecolor{dialinecolor}{rgb}{0.000000, 0.000000, 0.000000}
\pgfsetstrokecolor{dialinecolor}
\draw (21.619903\du,2.466107\du)--(21.621496\du,1.966109\du);
\pgfsetlinewidth{0.100000\du}
\pgfsetdash{}{0pt}
\pgfsetmiterjoin
\pgfsetbuttcap
\definecolor{dialinecolor}{rgb}{0.000000, 0.000000, 0.000000}
\pgfsetfillcolor{dialinecolor}
\pgfpathmoveto{\pgfpoint{28.337500\du}{2.237500\du}}
\pgfpathcurveto{\pgfpoint{28.337261\du}{2.312500\du}}{\pgfpoint{28.262023\du}{2.387260\du}}{\pgfpoint{28.187023\du}{2.387022\du}}
\pgfpathcurveto{\pgfpoint{28.112023\du}{2.386783\du}}{\pgfpoint{28.037263\du}{2.311544\du}}{\pgfpoint{28.037502\du}{2.236545\du}}
\pgfpathcurveto{\pgfpoint{28.037740\du}{2.161545\du}}{\pgfpoint{28.112979\du}{2.086784\du}}{\pgfpoint{28.187978\du}{2.087023\du}}
\pgfpathcurveto{\pgfpoint{28.262978\du}{2.087262\du}}{\pgfpoint{28.337739\du}{2.162500\du}}{\pgfpoint{28.337500\du}{2.237500\du}}
\pgfusepath{fill}
\definecolor{dialinecolor}{rgb}{0.000000, 0.000000, 0.000000}
\pgfsetstrokecolor{dialinecolor}
\draw (28.213297\du,1.987103\du)--(28.211704\du,2.487101\du);
\definecolor{dialinecolor}{rgb}{0.000000, 0.000000, 0.000000}
\pgfsetstrokecolor{dialinecolor}
\node[anchor=west] at (22.950000\du,1.500000\du){$s_-$};
\definecolor{dialinecolor}{rgb}{0.000000, 0.000000, 0.000000}
\pgfsetstrokecolor{dialinecolor}
\node[anchor=west] at (31.475000\du,1.540000\du){$s_+$};
\definecolor{dialinecolor}{rgb}{0.000000, 0.000000, 0.000000}
\pgfsetstrokecolor{dialinecolor}
\node[anchor=west] at (18.450000\du,5.735000\du){$t_-$};
\definecolor{dialinecolor}{rgb}{0.000000, 0.000000, 0.000000}
\pgfsetstrokecolor{dialinecolor}
\node[anchor=west] at (24.325000\du,5.780000\du){$t_+$};
\definecolor{dialinecolor}{rgb}{0.000000, 0.000000, 0.000000}
\pgfsetstrokecolor{dialinecolor}
\node[anchor=west] at (17.950000\du,6.750000\du){};
\pgfsetlinewidth{0.050000\du}
\pgfsetdash{{\pgflinewidth}{0.200000\du}}{0cm}
\pgfsetdash{{\pgflinewidth}{0.200000\du}}{0cm}
\pgfsetbuttcap
{
\definecolor{dialinecolor}{rgb}{0.000000, 0.000000, 0.000000}
\pgfsetfillcolor{dialinecolor}
\definecolor{dialinecolor}{rgb}{0.000000, 0.000000, 0.000000}
\pgfsetstrokecolor{dialinecolor}
\draw (19.010400\du,4.730400\du)--(21.660400\du,2.130350\du);
}
\definecolor{dialinecolor}{rgb}{0.000000, 0.000000, 0.000000}
\pgfsetstrokecolor{dialinecolor}
\node[anchor=west] at (16.750000\du,2.200000\du){$S$};
\definecolor{dialinecolor}{rgb}{0.000000, 0.000000, 0.000000}
\pgfsetstrokecolor{dialinecolor}
\node[anchor=west] at (16.800000\du,4.900000\du){$T$};
\pgfsetlinewidth{0.050000\du}
\pgfsetdash{{\pgflinewidth}{0.200000\du}}{0cm}
\pgfsetdash{{\pgflinewidth}{0.200000\du}}{0cm}
\pgfsetbuttcap
{
\definecolor{dialinecolor}{rgb}{0.000000, 0.000000, 0.000000}
\pgfsetfillcolor{dialinecolor}
\definecolor{dialinecolor}{rgb}{0.000000, 0.000000, 0.000000}
\pgfsetstrokecolor{dialinecolor}
\draw (25.610400\du,4.745400\du)--(28.260400\du,2.145350\du);
}
\pgfsetlinewidth{0.100000\du}
\pgfsetdash{}{0pt}
\pgfsetdash{}{0pt}
\pgfsetbuttcap
{
\definecolor{dialinecolor}{rgb}{0.000000, 0.000000, 0.000000}
\pgfsetfillcolor{dialinecolor}
}
\definecolor{dialinecolor}{rgb}{0.000000, 0.000000, 0.000000}
\pgfsetstrokecolor{dialinecolor}
\draw (19.145699\du,4.705697\du)--(25.487501\du,4.725903\du);
\pgfsetlinewidth{0.100000\du}
\pgfsetdash{}{0pt}
\pgfsetmiterjoin
\pgfsetbuttcap
\definecolor{dialinecolor}{rgb}{0.000000, 0.000000, 0.000000}
\pgfsetfillcolor{dialinecolor}
\pgfpathmoveto{\pgfpoint{18.895700\du}{4.704900\du}}
\pgfpathcurveto{\pgfpoint{18.895939\du}{4.629900\du}}{\pgfpoint{18.971178\du}{4.555140\du}}{\pgfpoint{19.046177\du}{4.555379\du}}
\pgfpathcurveto{\pgfpoint{19.121177\du}{4.555618\du}}{\pgfpoint{19.195937\du}{4.630856\du}}{\pgfpoint{19.195698\du}{4.705856\du}}
\pgfpathcurveto{\pgfpoint{19.195460\du}{4.780856\du}}{\pgfpoint{19.120221\du}{4.855616\du}}{\pgfpoint{19.045221\du}{4.855377\du}}
\pgfpathcurveto{\pgfpoint{18.970222\du}{4.855138\du}}{\pgfpoint{18.895461\du}{4.779900\du}}{\pgfpoint{18.895700\du}{4.704900\du}}
\pgfusepath{fill}
\definecolor{dialinecolor}{rgb}{0.000000, 0.000000, 0.000000}
\pgfsetstrokecolor{dialinecolor}
\draw (19.019903\du,4.955297\du)--(19.021496\du,4.455300\du);
\pgfsetlinewidth{0.100000\du}
\pgfsetdash{}{0pt}
\pgfsetmiterjoin
\pgfsetbuttcap
\definecolor{dialinecolor}{rgb}{0.000000, 0.000000, 0.000000}
\pgfsetfillcolor{dialinecolor}
\pgfpathmoveto{\pgfpoint{25.737500\du}{4.726700\du}}
\pgfpathcurveto{\pgfpoint{25.737261\du}{4.801700\du}}{\pgfpoint{25.662022\du}{4.876460\du}}{\pgfpoint{25.587023\du}{4.876221\du}}
\pgfpathcurveto{\pgfpoint{25.512023\du}{4.875982\du}}{\pgfpoint{25.437263\du}{4.800744\du}}{\pgfpoint{25.437502\du}{4.725744\du}}
\pgfpathcurveto{\pgfpoint{25.437740\du}{4.650744\du}}{\pgfpoint{25.512979\du}{4.575984\du}}{\pgfpoint{25.587979\du}{4.576223\du}}
\pgfpathcurveto{\pgfpoint{25.662978\du}{4.576462\du}}{\pgfpoint{25.737739\du}{4.651700\du}}{\pgfpoint{25.737500\du}{4.726700\du}}
\pgfusepath{fill}
\definecolor{dialinecolor}{rgb}{0.000000, 0.000000, 0.000000}
\pgfsetstrokecolor{dialinecolor}
\draw (25.613297\du,4.476303\du)--(25.611704\du,4.976300\du);
\end{tikzpicture}

%% file: sec-other.tex
This section contains the proof of~\autoref{thm:hierarchy_nonadaptive}. Note that the corresponding lower bound follows from~\autoref{cor:max_diff_lower_bound}, so we only need to prove the upper bound.

 Following the notation of \cite{NewmanRabinovich2017}, for a set $\mathcal{A}$ of disjoint $\pi$-copies in $f\colon [n] \to \R$ we define $T_i = T_i(\mathcal{A}) = \{ t_i : (t_1, \ldots, t_k) \in \mathcal{A}\}$ for any $1 \leq i \leq k$. We also define 
 \begin{align*}
 T^*(\mathcal{A}) &= \{(t_1, \ldots, t_k) : \forall i\ t_i \in T_i(\mathcal{A})\ ,\ \forall i \neq j\  f(t_i) < f(t_j) \iff \pi(i) < \pi(j) \}\\
 T^*_{i,j}(\mathcal{A}) &= \{(t_i, \ldots, t_j) : (t_1, \ldots, t_k) \in T^*(\mathcal{A}) \}
 \end{align*}
 That is, $\mathcal{T}^*$ is the set of all $\pi$-copies induced by copies from $\mathcal{A}$, where the entry $t_i$ from a copy $(t_1, \ldots, t_k) \in \mathcal{A}$ is only allowed to play the role of an $i$-th entry of a copy. $\mathcal{T}^*_{i,j}$ is the projection of $\mathcal{T}^*$ onto coordinates $i$ to $j$. 
 
 The proof of~\autoref{thm:hierarchy_nonadaptive} uses the \textsf{DyadicSampler} (Algorithm 3.1 of~\cite{NewmanRabinovich2017}) to efficiently find a large set of monotone subsequences of a desired form, and combines it with uniform sampling of the entries of the input sequence, to show that a $\pi$-copy can be obtained with good probability.
 
 \begin{proof}[Proof of~\autoref{thm:hierarchy_nonadaptive}]
 	Let $2 \leq \ell \leq k-1$. (the case $\ell=1$ corresponds to a monotone permutation, and is settled by Newman et al.)
 	Take $\pi = (\pi_1, \ldots, \pi_k)$ to be any permutation in which $\pi_{\ell} = 1$ and $\pi_{\ell+i} = \ell+i$ for any $1 \leq i \leq k-\ell$.
 	
 	The lower bound follows from~\autoref{thm:main_thm_lower2} and it remains to obtain the corresponding upper bound. 
 	Let $f\colon [n] \to \R$ be $\eps$-far from $\pi$-freeness, so $f$ contains a set $\mathcal{A}$ of $\eps n / k$ pairwise-disjoint $\pi$-copies.
 	Our algorithm for finding a $\pi$-copy in $f$ with probability $2/3$ is described below. To simplify the presentation, we do not try to optimize the (polynomial) dependence of the number of queries in $\eps$ and $\log{n}$.
 	\begin{itemize}
 		\item We run the \textsf{DyadicSampler} sufficiently many times, where each run is independent of all other runs. Our goal here is to obtain a set $T' \subseteq T^*_{\ell,k}$ of $n^{1 - 1 / \ell}$
 		monotone increasing subsequences of $f$, that are $\ell$-dominating in the following sense. A sequence $s = (s_1, \ldots, s_{k-\ell+1}) \in T^*_{\ell, k}$ is \emph{$\ell$-dominating} if the unique $t(s) = (t_1, \ldots, t_k) \in \mathcal{A}$ for which $s_1 = t_{\ell}$ also satisfies $f(t_{\ell+1}) \leq f(s_{2})$. We also require that the subsequences of $T'$ have disjoint first entries. That is, for any two tuples $s = (s_1, \ldots, s_{k-\ell+1}), s' = (s'_1, \ldots, s'_{k-\ell+1}) \in T'$, it holds that $s_1 \neq s'_1$.
 		Later, we show that after (independently) running the dyadic sampler $\tilde{\Theta}_{\eps}(n^{1 - 1 / \ell})$ times (making $\tilde{\Theta}_{\eps}(n^{1 - 1 / \ell})$ queries in total, since each run makes $k$ queries), a set $T'$ of the required size is obtained with probability $9/10$, provided that the constants hidden in the $\tilde{\Theta}_{\eps}$ term are large enough.
 		
 		\item Suppose now that a set $T'$ with the desired size was obtained in the first step.
 		The second step is to sample single entries of our input sequence $f$ uniformly at random, where the probability for each entry to be sampled is $10 n^{-1/{\ell}}$, independently of other samples. 
 		The crucial idea here is that the set of entries of a sequence $s \in T'$, and the first $\ell-1$ entries of the $\ell$-dominated $\pi$-copy $t = t(s) \in \mathcal{A}$, can be combined together to obtain a $\pi$-copy.
 		For a single subsequence $s \in T'$, Let $E_s$ be the event that all of the first $\ell-1$ entries of $t(s)$ are sampled. Then $\Pr(E_s) = 10^{\ell-1} n^{-1 + 1/\ell}$, and $E_s$ is independent of all other events $E_{s'}$ for $s' \in T'$. Thus, it is not hard to see that with probability at least $9/10$, there exists $s \in T'$ for which $E_s$ holds.
 		
 		\item The number of entries sampled in the second step is bounded by $100 n^{1-1/{\ell}}$ with probability at least $9/10$ (so if the number of proposed samples exceeds this, we may stop and return an arbitrary answer, similarly to what was done in~\autoref{sec:upper}). Thus, the probability that a $\pi$-copy is found using our test, which makes $\tilde{\Theta}_{\eps}(n^{1 - 1/\ell})$ queries, is at least $7/10$.
 	\end{itemize}
 	
 	It remains to show that a set $T'$ of $n^{1 - 1 / \ell}$ monotone increasing subsequences as above may indeed be produced with probability $9/10$ using $\tilde{\Theta}_{\eps}(n^{1-1/\ell})$ runs of the dyadic sampling algorithm. The fact that the dyadic algorithm essentially generates $\ell$-dominating monotone subsequences at no additional cost is a direct consequence of the proof of Theorem 3.2 in \cite{NewmanRabinovich2017} (as was first observed by Newman et al., see the beginning of Section 5 in their paper). 
 	
 	It was shown in \cite{NewmanRabinovich2017} that the success probability of the dyadic sampler (i.e., the probability to generate an $\ell$-dominating monotone subsequence as defined above) is at least $1 / \beta_k(n, \eps)$, where $\beta_k$ is polynomial in $\log{n}$ and $\eps$.
 	
 	First, we make $g_{k,l}(n, \eps) = 20 \alpha_k \beta_k(n, \eps) n^{1 - 1 / \ell} \log{n}$ mutually independent runs of the dyadic sampler (where $\alpha_k$ is determined in~\autoref{lem:dyadic_distribution}), and denote by $S$ the set of output tuples of all successful runs.
 	By Chernoff's bound, $\Pr( |S| \geq 10 \alpha_k n^{1 - 1 / \ell} \log{n}) > 99/100$ for large enough $n$. 
 	The next lemma can be used to show that $T'$ is large enough with good probability, assuming that $S$ is large enough.
 	
 	\begin{lemma}
 		\label{lem:dyadic_distribution}
 		Let $s = (s_1, \ldots, s_{r})$ be the output of a single run of the \textsf{DyadicSampler} of Newman et al. (Algorithm 3.1 in \cite{NewmanRabinovich2017}) with parameters $I$ and $r$, where $I$ is a sequence of length $m \geq r$.
 		Then for any element $x \in I$, the probability that $s_1$ equals $x$
 		is bounded by $\alpha_r / m$, where $\alpha_r$ depends only on $r$. 
 	\end{lemma}
 	\begin{proof}
 	For $r=1$ this is obvious, and $\alpha_1 = 1$ suffices. The rest of the proof is by induction.
 	If the ``split point'' is chosen to be $\ell \in [r-1]$ in step 2 of the sampler, and the ``slice-width'' is chosen to be $W = 2^w$ in step $3$ of the algorithm, then the \textsf{DyadicSampler} recursively takes (in step 5) its first $\ell$ elements from an interval $I_L$ of size at most $2\ell W < 2^{w+1}r$, chosen uniformly (in step 4) among a set of $\Theta(m / 2^w)$ possible intervals, such that each interval has a shift of $2^w$ from its predecessor. Therefore, each element is contained in at most $2r$ such intervals.
 	Note that an interval chosen to be $I_L$ in \textsf{DyadicSampler} must not be the last one -- otherwise $I_R$ would be empty -- so $|I_L| \geq 2^w$.
 	Combining the inductive assumption with all of the above considerations, the probability to choose a certain entry as $s_1$ is bounded by
 	\begin{align*}
 	\frac{1}{\log m} \sum_{w=0}^{\lfloor \log m - 1 \rfloor} \sum_{\ell=1}^{r-1} \left(O\left(\frac{ 2^{w}r}{m}\right)\frac{\alpha_\ell}{|I_L|}	\right) 
 	\leq \frac{r}{m \log{m}}  \sum_{w=0}^{\lfloor \log m - 1 \rfloor} \sum_{\ell=1}^{r-1} \alpha_{\ell} \leq \frac{1}{m} r \sum_{\ell=1}^{r-1} \alpha_{\ell}
 	\end{align*}
 	So picking $\alpha_r = r \sum_{\ell=1}^{r-1} \alpha_{\ell}$ suffices for our purposes.
 	\end{proof}
 	
 	With~\autoref{lem:dyadic_distribution} in hand, the rest of the proof is quite straightforward. We may assume that $n$ is large enough, so that $g_{k, l}(n, \eps) < n/10 \alpha_k$.
 	For a specific $x \in [n]$ the probability that among the $g_{k,l}(n, \eps)$ runs of the dyadic sampler, at least $\log{n}$ outputted a tuple in which $x$ is the first element is bounded by
 	\[
 	\binom{n/10\alpha_k}{\log{n}} \left(\frac{\alpha_k}{n} \right)^{\log{n}} \leq \left(\frac{n}{10 \alpha_k}\right)^{\log{n}} \left(\frac{\alpha_k}{n} \right)^{\log{n}} \leq \frac{1}{10}
 	\] 
 	Thus, with probability at least $89/100$, both $|S| \geq 10 \alpha_k n^{1 - 1 / \ell} \log{n}$ and no entry from $[n]$ appears as the first entry of more than $\log{n}$ tuples from $S$. In this case, $|T'| \geq n^{1 - 1/\ell}$ as required.
 \end{proof}

%% file: main-fop.bbl
\newcommand{\etalchar}[1]{$^{#1}$}
\begin{thebibliography}{CGG{\etalchar{+}}17}

\bibitem[BGJ{\etalchar{+}}12]{BGJRW:12}
Arnab Bhattacharyya, Elena Grigorescu, Kyomin Jung, Sofya Raskhodnikova, and
  David~P. Woodruff.
\newblock Transitive-closure spanners.
\newblock {\em SIAM Journal on Computing}, 41(6):1380--1425, 2012.

\bibitem[BY17]{BY:17}
Arnab Bhattacharyya and Yuichi Yoshida.
\newblock {\em Property Testing}.
\newblock Forthcoming, 2017.

\bibitem[Can15]{Canonne:15:Survey}
Cl{\'e}ment~L. Canonne.
\newblock A {S}urvey on {D}istribution {T}esting: your {D}ata is {B}ig. {B}ut
  is it {B}lue?
\newblock {\em Electronic Colloquium on Computational Complexity (ECCC)},
  22:63, April 2015.

\bibitem[CG17]{CG:17}
Cl{\'{e}}ment~L. Canonne and Tom Gur.
\newblock An adaptivity hierarchy theorem for property testing.
\newblock In {\em Computational Complexity Conference (CCC)}, 2017.

\bibitem[CGG{\etalchar{+}}17]{CGGKW:17}
Cl{\'{e}}ment~L. Canonne, Elena Grigorescu, Siyao Guo, Akash Kumar, and Karl
  Wimmer.
\newblock Testing $k$-monotonicity.
\newblock In {\em ITCS}, 2017.

\bibitem[CS13]{CS:2013}
Deeparnab Chakrabarty and C.~Seshadhri.
\newblock Optimal bounds for monotonicity and {L}ipschitz testing over
  hypercubes and hypergrids.
\newblock In {\em Proceedings of the Forty-fifth Annual ACM Symposium on Theory
  of Computing}, STOC '13, pages 419--428, New York, NY, USA, 2013. ACM.

\bibitem[DGL{\etalchar{+}}99]{Dodis:99}
Yevgeniy Dodis, Oded Goldreich, Eric Lehman, Sofya Raskhodnikova, Dana Ron, and
  Alex Samorodnitsky.
\newblock Improved testing algorithms for monotonicity.
\newblock In Dorit~S. Hochbaum, Klaus Jansen, Jos{\'e} D.~P. Rolim, and
  Alistair Sinclair, editors, {\em RANDOM-APPROX}, pages 97--108, Berlin,
  Heidelberg, 1999. Springer Berlin Heidelberg.

\bibitem[DKW56]{DKW:56}
Aryeh Dvoretzky, Jack Kiefer, and Jacob Wolfowitz.
\newblock Asymptotic minimax character of the sample distribution function and
  of the classical multinomial estimator.
\newblock {\em The Annals of Mathematical Statistics}, 27(3):642--669, 09 1956.

\bibitem[DN04]{DavidNagaraja:2004}
H.A. David and H.N. Nagaraja.
\newblock {\em Order Statistics}.
\newblock Wiley Series in Probability and Statistics. Wiley, 2004.

\bibitem[EKK{\etalchar{+}}00]{Ergun:00}
Funda Ergün, Sampath Kannan, S.Ravi Kumar, Ronitt Rubinfeld, and Mahesh
  Viswanathan.
\newblock Spot-checkers.
\newblock {\em Journal of Computer and System Sciences}, 60(3):717 -- 751,
  2000.

\bibitem[ES35]{ErdosSzekeres1935}
Paul Erd\H{o}s and George Szekeres.
\newblock A combinatorial problem in geometry.
\newblock {\em Compos.\@ Math.}, 2:463--470, 1935.

\bibitem[Fis04]{Fischer2004}
Eldar Fischer.
\newblock On the strength of comparisons in property testing.
\newblock {\em Inform. and Comput.}, 189:107--116, 2004.

\bibitem[GGR98]{GGR:98}
Oded Goldreich, Shafi Goldwasser, and Dana Ron.
\newblock Property testing and its connection to learning and approximation.
\newblock {\em Journal of the ACM}, 45(4):653--750, 1998.

\bibitem[GKW17]{GKW:17}
Elena Grigorescu, Akash Kumar, and Karl Wimmer.
\newblock {$K$-Monotonicity is Not Testable on the Hypercube}.
\newblock {\em Electronic Colloquium on Computational Complexity {(ECCC)}},
  24:88, 2017.

\bibitem[GM14]{GuillemotMarx2014}
Sylvain Guillemot and Daniel M\'arx.
\newblock Finding small patterns in permutations in linear time.
\newblock In {\em Proceedings of the Twenty-Fifth Annual ACM-SIAM Symposium on
  Discrete Algorithms}, pages 82--101, 2014.

\bibitem[Gol10]{Gol:10}
Oded Goldreich, editor.
\newblock {\em Property Testing - Current Research and Surveys [outgrow of a
  workshop at the Institute for Computer Science {(ITCS)} at Tsinghua
  University, January 2010]}, volume 6390 of {\em Lecture Notes in Computer
  Science}. Springer, 2010.

\bibitem[Gol17]{Gol:17}
Oded Goldreich.
\newblock {\em Introduction to Property Testing}.
\newblock Forthcoming, 2017.

\bibitem[JR13]{JhaR:13}
Madhav Jha and Sofya Raskhodnikova.
\newblock Testing and reconstruction of {L}ipschitz functions with applications
  to data privacy.
\newblock {\em {SIAM} J. Comput.}, 42(2):700--731, 2013.

\bibitem[Mas90]{Massart:90}
Pascal Massart.
\newblock The tight constant in the {D}voretzky--{K}iefer--{W}olfowitz
  inequality.
\newblock {\em The Annals of Probability}, 18(3):1269--1283, 07 1990.

\bibitem[New10]{Newman:10}
Ilan Newman.
\newblock Property testing of massively parametrized problems -- {A} survey.
\newblock In Oded Goldreich, editor, {\em Property Testing - Current Research
  and Surveys [outgrow of a workshop at the Institute for Computer Science
  {(ITCS)} at Tsinghua University, January 2010]}, volume 6390 of {\em Lecture
  Notes in Computer Science}, pages 142--157. Springer, 2010.

\bibitem[NRRS17]{NewmanRabinovich2017}
Ilan Newman, Yuri Rabinovich, Deepak Rajendraprasad, and Christian Sohler.
\newblock Testing for forbidden order patterns in an array.
\newblock In {\em Proceedings of the Twenty-Eighth Annual ACM-SIAM Symposium on
  Discrete Algorithms}, pages 1582--1597, 2017.

\bibitem[PRR03]{PRR:03}
Michal Parnas, Dana Ron, and Ronitt Rubinfeld.
\newblock On testing convexity and submodularity.
\newblock {\em SIAM Journal on Computing}, 32(5):1158--1184, 2003.

\bibitem[Ras14]{Raskhodnikova:14}
Sofya Raskhodnikova.
\newblock Testing if an array is sorted.
\newblock In Ming-Yang Kao, editor, {\em Encyclopedia of Algorithms}, pages
  1--5. Springer Berlin Heidelberg, Berlin, Heidelberg, 2014.

\bibitem[Ron08]{Ron:08}
Dana Ron.
\newblock Property testing: {A} learning theory perspective.
\newblock {\em Foundations and Trends in Machine Learning}, 1(3):307--402,
  2008.

\bibitem[Ron09]{Ron:09}
Dana Ron.
\newblock Algorithmic and analysis techniques in property testing.
\newblock {\em Foundations and Trends in Theoretical Computer Science},
  5(2):73--205, 2009.

\bibitem[RS96]{RS:96}
Ronitt Rubinfeld and Madhu Sudan.
\newblock Robust characterization of polynomials with applications to program
  testing.
\newblock {\em SIAM Journal on Computing}, 25(2):252--271, 1996.

\end{thebibliography}
